\documentclass[11pt, a4paper]{article}
\usepackage[utf8]{inputenc}
\usepackage[T1]{fontenc}
\usepackage{fullpage}
\usepackage{hyperref}
\usepackage{amsthm}
\usepackage{amssymb}
\usepackage{amsmath}
\usepackage{eqnarray}
\usepackage{enumerate}
\usepackage{graphicx} 
\usepackage{subcaption}
\usepackage{xcolor}
\usepackage{authblk}
\usepackage[authoryear,round]{natbib}
\usepackage{multirow}
\usepackage[font=small,labelfont=it,textfont=it]{caption}
\usepackage{booktabs}
\usepackage{float}
\usepackage{tkz-tab}
\usepackage{bbold}
\usepackage{tabularx}
\usepackage{empheq}
\usepackage{enumitem}
\usepackage{makecell}

\newtheorem{theorem}{Theorem}
\newtheorem{proposition}{Proposition}
\newtheorem{remark}{Remark}
\newtheorem{lemma}{Lemma}

\newtheorem{example}{Example}

\providecommand{\keywords}[1]
{
  \small	
  \textbf{Keywords:} #1
}
\newcolumntype{Y}{>{\centering\arraybackslash}X}

\newcommand{\customcite}[2][]{
  [\citeauthor{#2} (\citeyear{#2}), #1]
}

\DeclareMathOperator*{\esssup}{ess\,sup}

\title{Existence, uniqueness and positivity of solutions to the Guyon-Lekeufack path-dependent volatility model with general kernels}
\author[1,2]{Hervé Andrès}
\author[2]{Benjamin Jourdain}
\affil[1]{Milliman R\&D, Paris, France}
\affil[2]{CERMICS, École des Ponts, INRIA, Marne-la-Vallée, France.}
\date{\today}
\begin{document}

\maketitle

\begin{abstract}
    We show the existence and uniqueness of a continuous solution to a path-dependent volatility model introduced by \cite{guyon2023} to model the price of an equity index and its spot volatility. The considered model for the trend and activity features can be written as a Stochastic Volterra Equation (SVE) with non-convolutional and non-bounded kernels as well as non-Lipschitz coefficients. We first prove the existence and uniqueness of a solution to the SVE under integrability and regularity assumptions on the two kernels and under a condition on the second kernel weighting the past squared returns which ensures that the activity feature is bounded from below by a positive constant. Then, assuming in addition that the kernel weighting the past returns is of exponential type and that an inequality relating the logarithmic derivatives of the two kernels with respect to their second variables is satisfied, we show the positivity of the volatility process which is obtained as a non-linear function of the SVE's solution. We show numerically that the choice of an exponential kernel for the kernel weighting the past returns has little impact on the quality of model calibration compared to other choices and the inequality involving the logarithmic derivatives is satisfied by the calibrated kernels. These results extend those of \cite{nutz2023guyon}. 
\end{abstract}
\hspace{7pt}

\\

On a probability space $(\Omega, \mathcal{F}, \mathbb{P})$, we consider two independent Brownian motions $(W_t)_{t\ge 0}$ and $(B_t)_{t\ge 0}$ and let $\mathcal{F}_t=\sigma((B_s)_{s\ge 0},(W_s)_{s\in[0,t]})$ for $t\ge 0$. We consider a stock price $(S_t)_{t\ge 0}$ evolving for $t\ge 0$ with path-dependent volatility $\sigma_t$ according to
\begin{empheq}[left = \empheqlbrace]{align}
\frac{dS_t}{S_t} &= \sigma_t dW_t\label{eq:pdv_1}\\
  \sigma_t & = \beta_0 + \beta_1 R_{1,t} +\beta_2\sqrt{R_{2,t}}, \label{eq:pdv_2}\\
  R_{1,t} &=  \int_{-\Delta}^0 K_1(s,t)\sigma_s dB_{-s}+\int_{0}^t K_1(s,t)\sigma_s dW_s \label{eq:pdv_3}\\
  R_{2,t} &= \int_{-\Delta}^t K_2(s,t) \sigma_s^2 ds \label{eq:pdv_4}
\end{empheq}
where $\beta_0,\beta_2 \ge 0$, $\beta_1\le 0$, $\Delta \in [0,+\infty]$ and $K_1,K_2:\Gamma \rightarrow \mathbb{R}_+$ are kernels weighting the past returns and the past squared returns respectively with
\begin{equation*}
  \Gamma:= \{(s,t)\in \mathbb{R}\times\mathbb{R}^+ \mid s< t \}. 
\end{equation*}
We assume that $S_0=s_0>0$ and the past evolutions of $R_1$ and $R_2$ between $-\Delta$ and $0$ are given by the deterministic initial conditions $(r_{1,u})_{-\Delta \le u \le 0}$ and $(r_{2,u})_{-\Delta \le u \le 0}$ respectively with $r_{1,u} \in \mathbb{R}$ and $r_{2,u}\ge0$. The equality \eqref{eq:pdv_2} is assumed to hold also for $t\in (-\infty,0]$ when $\Delta=+\infty$ and for $t\in[-\Delta,0]$ otherwise. The integrand in the first term $\int_{-\Delta}^0 K_1(s,t)\sigma_s dB_{-s}$ in the right-hand side of \eqref{eq:pdv_3} is deterministic and this term is defined as the Wiener integral $\int_0^{\Delta}K_1(-s,t)\sigma_{-s} dB_{s}$.

This model has its roots in the work of \cite{guyon2023}. In this paper, they conduct an empirical study on the dependence of the volatility of an equity index with respect to the past path of the associated equity price. To this end, they consider the following model:
\begin{equation}\label{eq:empirical_model}
  \text{Volatility}_t = \beta_0+\beta_1 R_{1,t} + \beta_2 \sqrt{R_{2,t}}. 
\end{equation}
The definition and interpretation of each term is provided below.
\begin{itemize}
  \item $\text{Volatility}$ is a proxy measure of the spot volatility of an equity index. Two measures are considered: the value
  of an implied volatility index such as the VIX and an estimator of the realized volatility over one day using intraday observations of the equity index. 
  \item $R_1$ is a trend feature defined by:
  \begin{equation*}
    R_{1,t} = \sum_{t_i\le t} K_1(t-t_i) r_{t_i}
  \end{equation*}
  where $r_{t_i}$ is the daily return between day $t_{i-1}$ and day $t_i$ of the equity index and $K_1:\mathbb{R}_+\rightarrow \mathbb{R}_+$ is a decreasing kernel weighting the past returns. Since $\beta_1\le 0$, this feature allows to capture the leverage effect, i.e. the fact that volatility tends to rise when prices fall.
  \item $R_2$ is an activity or volatility feature defined by:
  \begin{equation*}
    R_{2,t} = \sum_{t_i \le t } K_2(t-t_i) r_{t_i}^2 
  \end{equation*}
  where $K_2$ is also a decreasing kernel. Since $\beta_2\ge 0$, this feature allows to capture the volatility clustering phenomenon, i.e. the fact that periods of large volatility tend to be followed by periods of large volatility, and periods of small volatility tend to be followed by periods of small volatility. Note that Guyon and Lekeufack consider the square root of $R_2$ in model (\ref{eq:empirical_model}) in order to get a quantity that is homogeneous to a volatility. 
\end{itemize}
The key result of Guyon and Lekeufack is to show that model (\ref{eq:empirical_model}), when calibrated on historical data, allows explaining (in the sense of the $R^2$ score) more than 80\% of the variations of the volatility for several main equity indices. They conclude that volatility is mostly path-dependent and they propose a continuous version (see model (4.1) in \citeauthor{guyon2023}, \citeyear{guyon2023}) of model (\ref{eq:empirical_model}). The model (\ref{eq:pdv_1})-(\ref{eq:pdv_4}) is a reproduction of this continuous version but with two differences:
\begin{enumerate}
  \item we consider general kernels for the weighting of the past returns and the past squared returns instead of convolutional ones and
  \item the lower boundary of the integrals defining $R_1$ and $R_2$ is parametrized by $\Delta$ instead of being fixed to $-\infty$. 
\end{enumerate}
Note that Equations (\ref{eq:pdv_2})-(\ref{eq:pdv_4}) can be rewritten as follows:
\begin{equation}\label{eq:r_dynamics}
  R_t = g(t)+ \int_0^t K_1(s,t) \gamma(R_s) dW_s+\int_0^t K_2(s,t)b(R_s)ds, \quad t\ge 0
\end{equation}
where $R_t := (R_{1,t},R_{2,t})$ and 
\begin{equation}\label{eq:g_def}
  \begin{aligned}
    g(t) &=\begin{pmatrix}
      g_1(t) \\
      g_2(t)
    \end{pmatrix}
      :=
      \begin{pmatrix}
      \int_{-\Delta}^0  K_1(s,t)(\beta_0+\beta_1r_{1,s}+\beta_2\sqrt{r_{2,s}})dB_{-s} \\
      \int_{-\Delta}^0 K_2(s,t) (\beta_0+\beta_1r_{1,s}+\beta_2\sqrt{r_{2,s}})^2 ds 
    \end{pmatrix} \\
    b(r_1,r_2) &= \begin{pmatrix}
      0 \\
      (\beta_0+\beta_1r_1+ \beta_2\sqrt{r_2})^2
    \end{pmatrix} \\
    \gamma(r_1,r_2) &= \begin{pmatrix}
      \beta_0+\beta_1r_1+ \beta_2\sqrt{r_2}\\
      0
    \end{pmatrix}. 
  \end{aligned}
\end{equation}
Thus, in the model defined by Equations (\ref{eq:pdv_1})-(\ref{eq:pdv_4}), the process $(R_1,R_2)$ is autonomous since $\sigma$ is a function of $R_1$ and $R_2$ only. Moreover, provided that Equation (\ref{eq:r_dynamics}) admits an $(\mathcal{F}_t)$-adapted solution such that $(\sigma_t)_{t\ge 0}$ defined by (\ref{eq:pdv_2}) satisfies: 
\begin{equation*}
  \mathbb{P}\left(\int_0^T \sigma_s^2 ds <+\infty \right) = 1, \quad \forall T>0,
\end{equation*}
then Equation (\ref{eq:pdv_1}) has a unique solution given by:
\begin{equation*}
  S_t = s_0 \exp\left(\int_0^t \sigma_s dW_s -\frac{1}{2}\int_0^t \sigma_s^2 ds \right). 
\end{equation*}

The objective of this note is to show global existence and uniqueness of a continuous solution to Equation (\ref{eq:r_dynamics}) and to establish sufficient conditions under which the volatility process $(\sigma_t)_{t\ge 0}$ is positive. Note that Equation (\ref{eq:r_dynamics}) corresponds to a Stochastic Volterra Equation (SVE). SVEs have been extensively studied in the literature starting with the work of \citeauthor{berger1} (\citeyear{berger1}, \citeyear{berger2}) who develop existence and uniqueness results for SVEs driven by a Brownian motion and with Lipschitz coefficients. These results have then been extended in numerous directions. \cite{protter1985} generalises the existence and uniqueness results to SVEs driven by right continuous semimartingales. \cite{pardoux1990} and \cite{alos1997} consider the case of SVEs with anticipating coefficients. \cite{cochran1995} and \cite{coutin2001} study SVEs with singular kernels, i.e. unbounded kernels such as the fractional kernel $K_H(s,t) = (t-s)^{H-1/2}$ for $H\in(0,1/2)$ and $0\le s \le t$. \cite{wang2008} proves the existence and uniqueness of a solution to an SVE with non-Lipschitz coefficients satisfying a condition similar to the one of \cite{yamada1971uniqueness} for stochastic differential equations. \cite{zhang2010stochastic} works on SVEs in Banach spaces with locally Lipschitz coefficients and singular kernels. More recently, there has been a resurgence of interest for SVEs in the mathematical finance literature with the advent of rough volatility models (see e.g. \citeauthor{gatheral2018volatility}, \citeyear{gatheral2018volatility} and \citeauthor{eleuch2018}, \citeyear{eleuch2018}) who rely on such equations. Among this new wave of publications, let us mention for instance those of \cite{abijaber2019}, \cite{abijaber2021} and \cite{jaber2024polynomial} who address several challenges raised by rough volatility models that could not be dealt with using the earlier results on SVEs. Due to the quadratic growth induced by the square and the fact that the square root function is non-differentiable at 0, the coefficients $b$ and $\gamma$ in Equation (\ref{eq:r_dynamics}) are non-Lipschitz and none of the above results can be used directly (even those considering non-Lipschitz coefficients) to show the existence and uniqueness of a solution to (\ref{eq:r_dynamics}). To obtain the existence, we resort to a localization technique which exploits the particular structure of the coefficients $b$ and $\gamma$.\\

Regarding the positivity of the volatility process, the SVE literature only deals with the existence of non-negative solutions (see \citeauthor{alfonsi2023nonnegativity}, \citeyear{alfonsi2023nonnegativity} and \citeauthor{alfonsi2024non}, \citeyear{alfonsi2024non}), while we aim at obtaining the positivity of a non-linear function of the coordinates of the solution of an SVE. In the path-dependent volatility (PDV) literature, \cite{guyon2023} first observed that, in the 2-factor Markovian PDV model (corresponding to model (\ref{eq:pdv_1})-(\ref{eq:pdv_4}) with exponential kernels, i.e. $K_i(s,t) = \lambda_i e^{-\lambda_i(t-s)}$ for $i\in\{1,2\}$), the volatility process $\sigma$ is non-negative provided that $\lambda_2<2\lambda_1$ because, under this condition, the drift of $\sigma$ is positive whenever $\sigma$ reaches 0. This result has then been strengthened by \cite{nutz2023guyon} who proved that $\sigma_t \ge \sigma_0\exp\left(\beta_1\lambda_1W_t -\lambda_1t -\frac{1}{2}\beta_1^2\lambda_1^2 t \right)$ so that the volatility process is bounded from below by a positive continuous stochastic process when its initial value $\sigma_0$ is positive. Moreover, they showed existence and uniqueness for the 4-factor Markovian PDV model which essentially corresponds to (\ref{eq:pdv_1})-(\ref{eq:pdv_4}) where both kernels are convex combinations of two exponential kernels, i.e. $K_i(s,t)=\theta_i\lambda_{1,i}e^{-\lambda_{1,i}(t-s)}+(1-\theta_i)\lambda_{2,i}e^{-\lambda_{2,i}(t-s)}$ for $i\in \{1,2\}$. This implies in particular the existence and uniqueness for the 2-factor Markovian PDV model. In fact, in the 4-factor Markovian PDV model, the trend and volatility features are both composed of two factors:
\begin{equation*}
  \begin{aligned}
    R_{1,t} &=  (1-\theta_1)R_{1,0,t} + \theta_1 R_{1,1,t}\\
    R_{2,t} &=  (1-\theta_2)R_{2,0,t} +   \theta_2R_{2,1,t}
  \end{aligned}
\end{equation*}
where $\theta_1,\theta_2 \in [0,1]$ and for $j\in\{0,1\}$:
\begin{equation*}
  \begin{aligned}
    R_{1,j,t} &= e^{-\lambda_{1,j}t}R_{1,j,0}+\int_{0}^t \lambda_{1,j} e^{-\lambda_{1,j}(t-s)}\sigma_s dW_s \\
    R_{2,j,t} &=e^{-\lambda_{2,j}t}R_{2,j,0} +\int_{0}^t \lambda_{2,j} e^{-\lambda_{2,j}(t-s)}\sigma_s^2 ds. 
  \end{aligned}
\end{equation*}
\cite{nutz2023guyon} show existence and uniqueness when the initial conditions satisfy: $R_{1,j,0} \in \mathbb{R}$ and $R_{2,j,0} >0$ for $j \in \{0,1\}$. Note that when $\lambda_{2,0} \ge \lambda_{2,1}$(which we can assume without loss of generality because of the symmetrical roles of $R_{2,0}$ and $R_{2,1}$), the positivity conditions can be weakened to $\theta_2R_{2,1,0}-(1-\theta_2) R_{2,0,0}^- >0$  since
$$\forall t\ge 0,\;R_{2,t}\ge (1-\theta_2)R_{2,0,0}e^{-\lambda_{2,0} t}+\theta_2R_{2,1,0}e^{-\lambda_{2,1} t}\ge \left(\theta_2R_{2,1,0}-(1-\theta_2) R_{2,0,0}^-\right)e^{-\lambda_{2,1} t}.$$
Our contribution is to extend these results to general kernels assuming only some integrability and regularity of the two kernels, in particular the kernels are not required to be of convolutional type or bounded. First, we show existence and uniqueness of a solution to the SVE (\ref{eq:r_dynamics}) and then under some additional assumptions on the two kernels, we find a positive stochastic lower bound for the volatility $\sigma$ by adapting the proof of \cite{nutz2023guyon} for the 2-factor PDV model.  \\

Let us now introduce the assumptions of our first theorem.      
\begin{enumerate}[label=(\textbf{I.\arabic*})] 
  \item For any $T>0$,
  \begin{equation*}
    \sup_{t\in[0,T]}\int_{-\Delta}^t \left(K_1(s,t)^2+K_2(s,t)\right)ds < \infty.
  \end{equation*} \label{hyp:1.1}
  \item For any $T>0$,
  \begin{equation*}
    \limsup_{\varepsilon \downarrow 0} \sup_{t\in [0,T]} \int_t^{t+\varepsilon} \left( K_1(s,t+\varepsilon)^2 + K_2(s,t+\varepsilon) \right) ds < 1.
  \end{equation*}\label{hyp:1.2}
  \item \begin{equation*}
    \sup_{s\in(-\Delta,0]} |\beta_0+\beta_1r_{1,s}+\beta_2\sqrt{r_{2,s}} | < \infty. 
  \end{equation*}\label{hyp:1.3}
  \item There exists $\alpha_1>1$ and $\alpha_2>1$ such that for any $T>0$,
  \begin{equation*}
    \sup_{t\in[0,T]} \int_0^t(K_1(s,t)^{2\alpha_1}+ K_2(s,t)^{\alpha_2} )ds <\infty.
  \end{equation*}\label{hyp:1.4}
  \item There exists $\gamma >0$ such that for each $T>0$, there exists a finite constant $C_T$ such that,
  \begin{equation*}
   \forall 0\le t < t'\le T,\;\sqrt{\int_{-\Delta}^t (K_1(s,t')-K_1(s,t))^2ds}+ \int_{-\Delta}^t |K_2(s,t')-K_2(s,t) | ds \le C_T(t'-t)^{\gamma}. 
  \end{equation*}\label{hyp:1.5}
  \item For any $T>0$,
  \begin{equation*} 
    \inf_{t\in [0,T]} g_2(t) >0. 
  \end{equation*}\label{hyp:1.6}
\end{enumerate}

Our first main result is the following:
\begin{theorem}\label{thm:main_result_1}
Under the assumptions \ref{hyp:1.1}-\ref{hyp:1.6}, there exists a unique continuous solution to Equation (\ref{eq:r_dynamics}). Moreover, the solution is locally $\gamma^*$-Hölder continuous for any $\gamma^*\in \left(0, \min\left(\gamma,\frac{1}{2\alpha_1^*},\frac{1}{\alpha_2^*} \right)\right)$ with $\alpha_i^* = \frac{\alpha_i}{\alpha_i-1}$. 
\end{theorem}
Assumptions \ref{hyp:1.1}-\ref{hyp:1.3} guarantee the existence and uniqueness of a local solution to Equation (\ref{eq:r_dynamics}) while assumptions \ref{hyp:1.4}-\ref{hyp:1.5} provide the regularity properties of the solution. Finally, \ref{hyp:1.6} ensures that $R_2$ is bounded from below by a positive deterministic function which allows dealing with the fact that the square root function is non-differentiable at 0 and in turn to deduce that the local solution is actually global. This assumption is satisfied whenever $s\mapsto \sigma^2_s \inf_{t\in [0,T]}K_2(s,t)$ is positive on a subset of $[-\Delta,0]$ with positive Lebesgue measure. When $K_2(s,t)=(1-\theta_2)\lambda_{2,0}e^{-\lambda_{2,0}(t-s)}+\theta_2\lambda_{2,1}e^{-\lambda_{2,1}(t-s)}$, the assumption is equivalent to the condition we gave above ($\theta_2 R_{2,1,0}-(1-\theta_2)R_{2,0,0}^- >0$ when $\lambda_{2,0} \ge \lambda_{2,1}$) in the case of the 4-factor Markovian PDV model. Remark that, by Hölder's inequality, assumption \ref{hyp:1.4} implies assumption \ref{hyp:1.2}, but we still write assumption \ref{hyp:1.2} to be able to separate the assumptions used to show the existence and uniqueness and those used to show the regularity of the solution.

\begin{remark}
   For homogeneous kernels $K_i(s,t)=\kappa_i(t-s)$ with $\kappa_i:{\mathbb R}^*_+\to{\mathbb R}_+$ for $i\in\{1,2\}$, \ref{hyp:1.1}, \ref{hyp:1.2} and \ref{hyp:1.4} hold as soon as the function $\kappa_1^{2\alpha}+\kappa_2^{\alpha}$ is locally integrable on $(0,+\infty)$ for some $\alpha>1$ and, when $\Delta=+\infty$, $\int_0^{+\infty}\left(\kappa_1^2(u)+\kappa_2(u)\right)du<\infty$. Moreover, when $\int_{-\Delta}^0{\mathbb 1}_{\{|\beta_0+\beta_1r_{1,s}+\beta_2\sqrt{r_{2,s}} |>0\}}ds>0$, then \ref{hyp:1.6} holds as soon as $\forall T>0$, $\inf_{t\in(0,T]}\kappa_2(t)>0$. In particular these conditions hold \begin{itemize}
   \item  when $\kappa_i(u)=\sum_{j=1}^{J_i}\theta_{ij}e^{-\lambda_{ij}u}$ with $J_1,J_2\ge 1$ and positive parameters $\theta_{ij}$ and $\lambda_{ij}$ or when $\kappa_i(u)=\frac{Z_{\zeta_i,\delta_i}}{(u+\delta_i)^{\beta_i}}$ with $\zeta_i>\frac i2\mathbb{1}_{\{\Delta=+\infty\}}$ and positive parameters $\delta_i,Z_{\zeta_i,\delta_i}>0$. For these choices, \ref{hyp:1.5} also holds with $\gamma=1$. Then the solution to (\ref{eq:r_dynamics}) is locally $\gamma^*$-Hölder continuous for any $\gamma^*\in \left(0,\frac  12\right)$, the limitation coming from $\frac 1{2\alpha_1^*}$.
   \item when $\Delta<\infty$ and $\kappa_i(u)=u^{H_i-\frac i2}$ with $H_1,H_2>0$. Let us consider the singular case when $H_1\in (0,\frac 12)$ and $H_2\in (0,1)$ and focus on assumption \ref{hyp:1.5}. For $0\le t<t'$, we have
\begin{align*}
   (K_1(s,t)-K_1(s,t'))^2\le K_1(s,t)\left(K_1(s,t)-K_1(s,t')\right)=(t-s)^{{H_1}-\frac 12}(1/2-{H_1})\int_t^{t'}(r-s)^{{H_1}-\frac 32}dr
\end{align*}
We deduce that when $-\Delta\le 2t-t'$
\begin{align*}
  \frac 2{1-2{H_1}}&\int_{-\Delta}^t(K_1(s,t)-K_1(s,t'))^2ds\le \int_{r=t}^{t'}\int_{s=-\Delta}^t(t-s)^{{H_1}-\frac 12}(r-s)^{{H_1}-\frac 32}ds dr\\
  &\le \int_{r=t}^{t'}\int_{s=-\Delta}^{2t-t'}(t-s)^{2{H_1}-2}ds dr+\int_{r=t}^{t'}\int_{s=2t-t'}^t(t-s)^{{H_1}-1}(r-t)^{{H_1}-1}ds dr\\
  &\le (t'-t)\times \frac{1}{1-2{H_1}}(t'-t)^{2{H_1}-1}+\frac{(t'-t)^{2{H_1}}}{{H_1}^2}.
\end{align*}
The left-hand side being increasing with $\Delta$, it remains smaller than the right-hand side when $\Delta>2t-t'$. On the other hand,
\begin{align*}
   \int_{-\Delta}^t K_2(s,t)-K_2(s,t')ds=\int_{r=t}^{t'}\int_{s=-\Delta}^t\frac{(r-s)^{H_2-2}}{2-H_2}ds dr\le \int_{t}^{t'}(r-t)^{H_2-1} dr=\frac{(t'-t)^{H_2}}{H_2}.\end{align*}Hence \ref{hyp:1.5} holds with $\gamma=\min(H_1,H_2)$. On the other hand, \ref{hyp:1.4} holds with $\alpha_1<\frac 1{1-2H_1}$ and $\alpha_2<\frac{1}{1-H_2}$, so that $\frac 1{2\alpha_1^*}=\frac 12\left(1-\frac 1{\alpha_1}\right)<H_1$ and $\frac 1{\alpha_2^*}=1-\frac 1{\alpha_2}<H_2$. Hence the solution to (\ref{eq:r_dynamics}) is locally $\gamma^*$-Hölder continuous for any $\gamma^*\in \left(0,\min(H_1,H_2)\right)$.
   \end{itemize}
\end{remark}
In order to get the positivity of the volatility $\sigma$, we now exclude singular kernels and assume that $K_1$ and $K_2$ are defined on:
\begin{equation*}
  \bar{\Gamma}:= \{(s,t)\in \mathbb{R}\times\mathbb{R}^+ \mid s\le t \}. 
\end{equation*}
We also need the following additional assumptions.
  
\begin{enumerate}[label=(\textbf{II.\arabic*})] 
    \item For all $s\in(-\Delta,+\infty)$, $t\mapsto K_1(s,t)$ and $t\mapsto K_2(s,t)$ are continuously differentiable on $[s,+\infty)$ and the two kernels verify for all $t\ge 0$:
    \begin{equation*}
      \int_0^t K_1(u,u)^2 du+ \int_0^t \left(\int_{-\Delta}^v \left|\partial_v K_1(u,v)\right|^2 du\right)^{1/2} dv + \int_0^t K_2(u,u)du+ \int_0^t \int_{-\Delta}^v \left|\partial_v K_2(u,v)\right| dudv  < \infty. 
    \end{equation*}\label{hyp:2.1} 
    \item $K_1(s,t) = f(s)e^{h(t)}$ where $h$ is assumed to be differentiable, non-increasing and such that $\inf_{t\ge 0}h'(t)>-\infty$.\label{hyp:2.2}
    \item $\partial_t K_2(s,t)-2h'(t)K_2(s,t) \ge 0$ for all $t\ge 0$ and $s\in (-\Delta,t]$.\label{hyp:2.3}
    \item $g_2(0)>0$.  \label{hyp:2.4}
  \end{enumerate}

We can now state our second main result. 
\begin{theorem}\label{thm:main_result_2}
  Under assumptions \ref{hyp:1.1}-\ref{hyp:1.5} and \ref{hyp:2.1}-\ref{hyp:2.4}, there exists a unique continuous solution to Equation (\ref{eq:r_dynamics}). This process has the same the same Hölder regularity as in Theorem \ref{thm:main_result_1}. Moreover, the process $(\sigma_t=\beta_0+\beta_1R_{1,t}+\beta_2\sqrt{R_{2,t}})_{t\ge 0}$ is continuous and bounded from below by the process $X$ defined by:
  \begin{equation*}
    X_t=\sigma_0\exp\left(\beta_1\int_0^t K_1(s,s) dW_s + h(t)-h(0) -\frac{1}{2}\beta_1^2\int_0^tK_1(s,s)^2ds  \right)
  \end{equation*}
  and thus is positive if $\sigma_0>0$. 
\end{theorem}
The assumption \ref{hyp:2.1} implies that $R_1$ and $R_2$ are Itô processes (see Lemma \ref{lem:differential_K}) while the assumptions \ref{hyp:2.2} and \ref{hyp:2.3} allow to use a comparison result for Itô processes (see Lemma \ref{lem:comparison_result}), similar to the one of \customcite[Proposition~2.18]{karatzas1991} for solutions of SDEs, which gives the positive stochastic lower bound in Theorem \ref{thm:main_result_2}. Finally, assumption \ref{hyp:2.4} combined with assumptions \ref{hyp:2.2} and \ref{hyp:2.3} implies assumption \ref{hyp:1.6} (see Lemma \ref{lem:condition_1.6}) which justifies that \ref{hyp:1.6} has not been included in the set of assumptions of Theorem \ref{thm:main_result_2}. This assumption is satisfied whenever $s\mapsto K_2(s,0)\sigma_s^2$ is positive on a subset of $[-\Delta,0]$ with positive Lebesgue measure. We provide below two examples of kernels $K_1$ satisfying both assumptions \ref{hyp:2.1} and \ref{hyp:2.2}. 

\begin{example}
  The exponential kernel $K_1(s,t)=\lambda e^{-\lambda (t-s)}$ with $\lambda>0$ is recovered by taking $f(s)=\lambda e^{\lambda s}$ and $h(t)=-\lambda t$.
\end{example}
\begin{example}
  By considering $f(s)=(s+\Delta)^{a}\mathbb{1}_{\{s\ge -\Delta\}}$ and $h(t)=-a\log (t+\Delta)$ for $a>0$ and $\Delta<\infty$, one obtains the non-convolutional kernel $K_1(s,t)=\left(\frac{s+\Delta}{t+\Delta}\right)^a\mathbb{1}_{\{s\ge -\Delta\}}$ which allows modeling long memory.   
\end{example}

Let us make two remarks regarding assumption \ref{hyp:2.3}. 

\begin{remark}\label{rk:nv_condition_equivalence}
   Let $K_1(u,t) = \lambda_1e^{-\lambda_1(t-u)}$. When $K_2(u,t) = \lambda_2e^{-\lambda_2(t-u)}$, condition \ref{hyp:2.3} is equivalent to $2\lambda_1 \ge \lambda_2$ so we recover the condition of \cite{nutz2023guyon}.\\ When $K_2(s,t)=(1-\theta_2)\lambda_{2,0}e^{-\lambda_{2,0}(t-s)}+{\theta_2}\lambda_{2,1}e^{-\lambda_{2,1}(t-s)}$, $(\textbf{II.3})$ is satisfied if and only if $$2\lambda_1\ge \lambda_{2,0}\frac{(1-\theta_2)\lambda_{2,0}}{(1-\theta_2)\lambda_{2,0}+\theta_2\lambda_{2,1}}+\lambda_{2,1}\frac{\theta_2\lambda_{2,1}}{(1-\theta_2)\lambda_{2,0}+\theta_2\lambda_{2,1}}.$$ The necessary condition is obtained by choosing $s=t$. The sufficient condition follows from the fact that the weight of the largest coefficient between $\lambda_{2,0}$ and $\lambda_{2,1}$ in the convex combination $$\frac{-\partial_t K_2(s,t)}{K_2(s,t)}=\lambda_{2,0}\frac{(1-\theta_2)\lambda_{2,0}e^{-\lambda_{2,0}(t-s)}}{(1-\theta_2)\lambda_{2,0}e^{-\lambda_{2,0}(t-s)}+\theta_2\lambda_{2,1}e^{-\lambda_{2,1}(t-s)}}+\lambda_{2,1}\frac{\theta_2\lambda_{2,1}e^{-\lambda_{2,1}(t-s)}}{(1-\theta_2)\lambda_{2,0}e^{-\lambda_{2,0}(t-s)}+\theta_2\lambda_{2,1}e^{-\lambda_{2,1}(t-s)}}$$ is non-increasing in $t-s$.
\end{remark}

\begin{remark}\label{rk:cutoff}
  Since the condition \ref{hyp:2.4} and the definition of $g_2(0)$ imply that $(-\Delta,0]\ni s\mapsto K_2(s,0)$ is positive on a set with positive Lebesgue measure and for $s\le 0\le t$, $K_2(s,t)\ge K_2(s,0)\exp\left(2\int_0^t h'(u)du\right)$, kernels vanishing for $t-s$ greater than some finite cutoff lag are excluded from our set of assumptions. This can be viewed as a limitation given that \cite{guyon2023} truncate the two kernels after some cut-off lag in their numerical experiments. However, we will show in the following that the $R^2$ scores achieved by the calibration of the PDV model (\ref{eq:empirical_model}) for several choices of kernels are still high (although a bit lower than those reported by \citeauthor{guyon2023}, \citeyear{guyon2023}) when the kernels are not truncated. 
\end{remark}

In their empirical study, \cite{guyon2023} choose to work with time-shifted power-law (TSPL) kernels for the two kernels $K_1$ and $K_2$:
\begin{equation*}
  K_i(u,t) = \frac{Z_{\zeta_i,\delta_i}}{(t-u+\delta_i)^{\zeta_i}}, \quad i=1,2,  
\end{equation*}
with $\zeta_i,\delta_i>0$ and $Z_{\zeta_i,\delta_i}$ a normalization constant. This choice is motivated by the ability of this parametric form to mix short and long memory, which is a characteristic of historical volatility data. However, by considering a TSPL kernel for $K_1$, assumption \ref{hyp:2.2} is not satisfied so we cannot apply Theorem \ref{thm:main_result_2}. Therefore, we propose to consider an exponential kernel for $K_1$ and a TSPL kernel for $K_2$, i.e.  $K_1(u,t) = \lambda e^{-\lambda(t-u)}$ and $K_2(u,t) = \frac{Z_{\zeta,\delta}}{(t-u+\delta)^{\zeta}}$. Assuming that $\zeta>1$ and $2\lambda \delta \ge \zeta$ (corresponding to condition \ref{hyp:2.3}), one can check that all the assumptions of Theorem \ref{thm:main_result_2} are satisfied. Note that if $\Delta<+\infty$, we can have $0<\zeta\le 1$. This alternative choice for $K_1$ and $K_2$ implying a guarantee of positivity for the volatility process, it is natural to compare its performance to the one of \cite{guyon2023}, that is two TSPL kernels, as well as two other options:
\begin{enumerate}
  \item the 2-factor Markovian PDV model (corresponding to exponential kernels for both $K_1$ and $K_2$) for which \cite{nutz2023guyon} gave a condition guaranteeing the positivity of the volatility;
  \item the 4-factor Markovian PDV model (corresponding to a convex combination of two exponential kernels), which was the option recommended by \cite{guyon2023} for applications (e.g. solving the joint SPX/VIX smile calibration problem) since it is Markovian and a convex combination of two exponential kernels provides a good approximation of a TSPL kernel. 
\end{enumerate}

To perform this comparison, we calibrate model (\ref{eq:empirical_model}) on implied volatility indices (the VIX for the S\&P 500, the VSTOXX for the Euro Stoxx 50 and the IVI for the FTSE 100), extracted from Refinitiv\footnote{\url{www.lseg.com}}, and realized volatilities computed using 5-minute returns per trading day coming from \cite{heber2009oxford}. The data of the corresponding indices is also extracted from Refinitiv. Note that the data is split into a train set spanning the period from January 1, 2000 to December 31, 2018 and a test set spanning the period from January 1, 2019 to May 15, 2022 (resp. to December 30, 2021) for the implied volatility indices (resp. for the realized volatilities). The calibration methodology is essentially the one described in \cite{guyon2023} but we make some modifications. First, unlike Guyon and Lekeufack who truncate the sums in the formulas of $R_1$ and $R_2$ in Equation (\ref{eq:empirical_model}) to 1,000 business days, we use all the past index data (which goes back to January 2, 1980 for the S\&P 500, December 31, 1986 for the Euro Stoxx 50 and January 3, 1984 for the FTSE 100). This is motivated by Remark \ref{rk:cutoff} which shows that kernels that vanish for $t-s>C$ for some cut-off lag $C>0$ do not satisfy the assumptions of Theorem \ref{thm:main_result_2}. Second, we bound the variables $\beta_0$ and $\beta_2$ from below by 0 and the variable $\beta_1$ from above by $0$ in the numerical optimization as this allows to stabilize the results for the VIX. For the same reason, we add an $L^2$ penalization in the objective function following \cite{andres2023implied} in the case of the two TSPL kernels. The resulting $R^2$ scores are presented in Table \ref{tab:pdv_model_calibration}. We observe that the compared choices for $K_1$ and $K_2$ provide very close results both on the train and the test sets. The best model in terms of $R^2$ scores is the convex combination of two exponential kernels for $K_1$ and $K_2$ but this is also the choice with the largest number of parameters (there is three parameters per kernel versus one for the exponential kernel and two for the TSPL kernel). Besides, there is no condition guaranteeing the positivity of the volatility available so far in this case (\citeauthor{nutz2023guyon}, \citeyear{nutz2023guyon}, actually show that the volatility can be negative with positive probability if one considers the analogue for the 4-factor PDV model of the positivity condition in the 2-factor PDV model). The same remark applies to the choice of two TSPL kernels which is the second best model on the train set. Finally, the choice of two exponential kernels is the one with the worst scores. In a talk given at Columbia University in May 2024, Julien Guyon has presented an analogous comparison of model $(2,2)$ with two convex combinations of two exponential kernels for $K_1$ and $K_2$, model $(1,2)$ with one exponential kernel for $K_1$ and a convex combination of two for $K_2$, model $(2,1)$ with a convex combination of two for $K_1$ and one for $K_2$ and finally model $(1,1)$ with one for $K_1$ and one for $K_2$. He showed that model $(1,2)$ performs slightly worse than model $(2,2)$ and better than model $(2,1)$ which itself is slightly better than model $(1,1)$. These results show that, in view of applications, the choice of an exponential kernel for $K_1$ and a TSPL kernel for $K_2$ is an attractive alternative to the choice of two TSPL kernels or two convex combinations of two exponential kernels since it does not deteriorate the fit to historical data, the model is parsimonious in terms of number of parameters and there is a sufficient condition ($2\lambda \delta \ge \zeta$) guaranteeing the positivity of the volatility. Moreover, this sufficient condition is satisfied numerically by the calibrated parameters as shown in Table \ref{tab:positivity_condition}.

\begin{table}[h]
  \centering
  \caption{$R^2$ scores of model (\ref{eq:empirical_model}) for various choices of kernels $K_1$ and $K_2$ }
  \label{tab:pdv_model_calibration}
  \setlength{\extrarowheight}{10pt}
  \begin{tabularx}{\textwidth}{clYYYYYY}
  \toprule
        & & \multicolumn{3}{c}{Implied volatility index} & \multicolumn{3}{c}{Realized volatility} \\ \cmidrule(lr){3-5} \cmidrule(l){6-8}
        & & VIX & VSTOXX & IVI & SPX & STOXX & FTSE \\ \midrule
        \multirow{2}{*}{$K_1$ and $K_2$ are TSPL kernels} &$R^2$ train    & 89.34\% & 91.36\% & 92.13\% & 67.03\% & 58.86\% & 61.74\% \\
        & $R^2$ test & 83.43\% & 90.91\% & 87.97\% & 68.84\% & 62.48\% & 62.72\% \\ \midrule
        \multirow{2}{*}{\makecell{$K_1$ and $K_2$ are exponential \\ kernels}} & $R^2$ train & 89.10\% & 90.28\% & 90.50\% & 66.99\% & 58.53\% & 61.54\% \\
        &$R^2$ test & 74.26\%& 86.43\% & 82.37\% & 63.41\% & 61.68\% & 62.48\% \\\midrule
        \multirow{2}{*}{\makecell{$K_1$ and $K_2$ are convex   \\ combinations of two \\exponential kernels}} & $R^2$ train &  90.01\% & 91.99\% & 92.37\% & 67.55\% & 59.10\% & 62.14\% \\
        & $R^2$ test  & 83.01\% & 91.61\%& 88.74\% & 68.86\%& 62.68\% & 62.80\%   \\\midrule
        \multirow{2}{*}{\makecell{$K_1$ is an exponential kernel\\ and $K_2$ is a TSPL kernel}} & $R^2$ train & 89.86\% & 91.76\% & 92.29\% & 67.12\% &  58.85\% & 61.74\% \\
        & $R^2$ test & 81.33\% & 91.03\% & 87.67\% & 64.61\% & 62.78\% & 62.88\% \\
  \bottomrule
  \end{tabularx}
  \end{table}

  \begin{table}[h]
    \centering
    \caption{Value of $\frac{2\lambda\delta}{\zeta}$ when we calibrate model (\ref{eq:empirical_model}) with an exponential kernel for $K_1$ and a TSPL kernel for $K_2$.}
    \label{tab:positivity_condition}
    \begin{tabular}{@{}ccccccc@{}}
    \toprule
     & VIX  & VSTOXX & IVI  & SPX  & STOXX & FTSE \\ \midrule
    $\frac{2\lambda\delta}{\zeta}$ & 4.00 & 2.84   & 1.62 & 2.23 & 1.97  & 1.97 \\ \bottomrule
    \end{tabular}
  \end{table}

The sequel is organized as follows. First, we introduce two theorems from \cite{zhang2010stochastic} regarding the existence and uniqueness of continuous solutions to a Stochastic Volterra Equation. In the next section, we prove Theorem \ref{thm:main_result_1}. The last section is dedicated to the proof of Theorem \ref{thm:main_result_2}. 

\section{Preliminary results}
We start by stating an existence and uniqueness result for a stochastic Volterra integral equation with (possibly) singular kernels and Lipschitz coefficients. This result is a simplified version of Theorem 3.1 of \cite{zhang2010stochastic} who works in a 2-smooth Banach space while we work in $\mathbb{R}^d$. We are interested in the following stochastic Volterra integral equation in $\mathbb{R}^d$:\\
\begin{equation}\label{eq:volterra}
  X_t = g(t) + \int_0^t K_1(s,t)\gamma(X_s)dW_s + \int_0^t K_2(s,t)b(X_s)ds
\end{equation}
where $(g(t))_{t\ge 0}$ is an $\mathbb{R}^d$-valued measurable and $(\mathcal{F}_t)$-adapted process, $K_1$ and $K_2$ are two kernels from $\Gamma$ to $\mathbb{R}_+$, $b$ is a function from $\mathbb{R}^d$ to $\mathbb{R}^d$, $\gamma$ is a function from $\mathbb{R}^d$ to $\mathbb{R}^d\times \mathbb{R}^q$ and $W$ is a $(\mathcal{F}_t)$-adapted $q$-dimensional Brownian motion.

\begin{theorem}\label{thm:global_existence}
  Assume that \ref{hyp:1.1} and \ref{hyp:1.2} as well as the following conditions hold: 
  \begin{enumerate}[label=(\roman*)]
    \item There exists $p\ge 2$ such that, for any $T>0$, 
    \begin{equation*}
      \sup_{t\in [0,T]} \mathbb{E}[\|g(t)\|^p] < \infty. 
    \end{equation*}\label{hyp:global_thm_1}
    \item There exists $C>0$ such that for all $x,y\in \mathbb{R}^d$,
    \begin{equation*}
      \|b(x)-b(y)\|+\|\gamma(x)-\gamma(y)\| \le C\|x-y\|.
    \end{equation*}\label{hyp:global_thm_2}
  \end{enumerate}
  Then, there exists a unique (up to $dt$ a.e. equality) $(\mathcal{F}_t)$-adapted process $(X_t)_{t\ge 0}$ such that:
  \begin{itemize}
    \item Equation (\ref{eq:volterra}) holds dt a.e. and
    \item $\forall T>0,\; \esssup_{t\in[0,T]} \mathbb{E}\left[\|X_t\|^p \right] < \infty. $
  \end{itemize}
\end{theorem}
\begin{remark}
  The conclusions of Theorem \ref{thm:global_existence} still hold if the integral is taken from 0 to $t$ in assumption \ref{hyp:1.1} instead of from $-\Delta$ to $t$. 
\end{remark}

Under additional conditions on the kernels $K_1$ and $K_2$, one can show that the solution to Equation (\ref{eq:volterra}) is continuous using the Kolmogorov-Centsov theorem. The following theorem is an adaptation of Theorem 3.3 in \cite{zhang2010stochastic}. 
\begin{theorem}\label{thm:continuity}
  Assume that the conditions of Theorem \ref{thm:global_existence} as well as conditions \ref{hyp:1.4} and \ref{hyp:1.5} hold. If $g$ is almost surely continuous and satisfies for any $T>0$ and any $p\ge 2$:
    \begin{equation*}
      \sup_{t\in [0,T]} \mathbb{E}\left[\|g(t)\|^p \right] <\infty,
    \end{equation*}
    then there exists a unique continuous solution to Equation (\ref{eq:volterra}). If, moreover,
\begin{equation*}
    \exists\delta>0,\;\forall p\ge 2,\;\forall T>0,\;\exists C_{p,T}<\infty,\;\mathbb{E}\left[\|g(t')-g(t)\|^p \right] \le C_{p,T} |t'-t|^{\delta p}\mbox{ for }0\le t < t' \le T,
  \end{equation*}
  then the process $(X_t)_{t\ge 0}$ is locally $\gamma^*$-Hölder continuous for any $\gamma^*\in \left(0, \min\left(\frac{1}{2\alpha_1^*},\frac{1}{\alpha_2^*},\gamma, \delta \right)\right)$ with $\alpha_i^* = \frac{\alpha_i}{\alpha_i-1}$. 
\end{theorem}
\begin{remark}
  The conclusion of Theorem \ref{thm:continuity} still holds if the integral is taken from 0 to $t$ in assumption \ref{hyp:1.5} instead of from $-\Delta$ to $t$.
\end{remark}


\section{Proof of Theorem \ref{thm:main_result_1}} 

The existence and uniqueness theorem (Theorem \ref{thm:global_existence}) cannot be used as is for model (\ref{eq:r_dynamics}) since the global Lipschitz condition is not satisfied because of the square root and the square. However, by using a localization technique, we can return to the assumptions of Theorem \ref{thm:global_existence} and show the existence of a solution up to some stopping time.

\paragraph{Step 1 of the proof of Theorem \ref{thm:main_result_1}: existence of a local solution.}

For $n\in \mathbb{N}^*$, we define, 
  \begin{equation*}
    \begin{aligned}
      b_n(r_1,r_2) &= \begin{pmatrix}
        0 \\
        \left(-n\vee \left(\beta_0+\beta_1r_1+\beta_2\sqrt{r_2\vee \frac{1}{n}}\right)\wedge n \right)\left(\beta_0+\beta_1r_1+\beta_2\sqrt{r_2\vee \frac{1}{n}} \right)
      \end{pmatrix},\\
      \gamma_n(r_1,r_2) &= \begin{pmatrix}
        \beta_0+\beta_1r_1+\beta_2\sqrt{r_2\vee \frac{1}{n}}\\
        0
      \end{pmatrix}.
    \end{aligned}
  \end{equation*}
  One can easily check that $b_n$ and $\gamma_n$ are globally Lipschitz.
Let us verify that $\sup_{t\in[0,T]}\mathbb{E}[\|g(t)\|^p]<\infty$ for all $p\ge 2$. Since the initial conditions $(r_{1,s})_{-\Delta\le s \le 0}$ and $(r_{2,s})_{-\Delta\le s \le 0}$ are deterministic, we have for all $p\ge 2$:
  \begin{equation*}
    \begin{aligned}
    \mathbb{E}\left[|g_1(t)|^p \right] &= \mathbb{E}\left[\left|G \right|^p \right] \left(\int_{-\Delta}^0K_1(s,t)^2 \left(\beta_0+\beta_1r_{1,s}+\beta_2\sqrt{r_{2,s}}\right)^2 ds  \right)^{\frac{p}{2}} \\
    &\le \mathbb{E}\left[\left|G \right|^p \right] \sup_{s\in (-\Delta,0]} \left|\beta_0+\beta_1 r_{1,s}+\beta_2 \sqrt{r_{2,s}} \right|^p \left( \int_{-\Delta}^0 K_1(s,t)^2ds \right)^{\frac{p}{2}}  \\ 
    \end{aligned}
  \end{equation*}
  where $G$ is a standard normal random variable. We deduce that $\mathbb{E}\left[|g_1(t)|^p \right]$ is finite by assumptions \ref{hyp:1.1} and \ref{hyp:1.3}. Moreover, still by assumptions \ref{hyp:1.1} and \ref{hyp:1.3}, we have for all $p\ge 2$:
  \begin{equation*}
    \mathbb{E}\left[|g_2(t)|^p \right] \le \sup_{s\in (-\Delta,0]} \left|\beta_0+\beta_1 r_{1,s}+\beta_2 \sqrt{r_{2,s}} \right|^{2p} \left( \int_{-\Delta}^0 K_2(s,t)ds \right)^{p} < \infty.
  \end{equation*}

Thus, Theorem \ref{thm:global_existence} guarantees the existence of a unique solution to the following equation :
  \begin{equation}\label{eq:local_solution}
    R_t^{(n)} = g(t) + \int_0^t K_2(s,t)b_n\left(R_s^{(n)}\right)ds + \int_0^t K_1(s,t) \gamma_n\left(R_s^{(n)}\right)dW_s, \quad \text{$dt$ a.e. on $[0,T]$}
  \end{equation}
verifying for all $p\ge 2$:
\begin{equation}\label{eq:boundedness_Rn}
  \esssup_{t\in[0,T]} \mathbb{E}\left[\left\|R_t^{(n)}\right\|^p \right] < \infty.
\end{equation} 
Moreover, using assumptions \ref{hyp:1.3} and \ref{hyp:1.5}, we have for any $p>0$ and $0\le t < t' \le T$: 
\begin{equation*}
  \begin{aligned}
    \mathbb{E}\left[\|g(t')-g(t)\|^p \right] & \le  C(|g_1(t')-g_1(t)|^p + |g_2(t')-g_2(t)|^p) \\
    &\le C \Biggl(\mathbb{E}\left[\left|\int_{-\Delta}^0 (K_1(s,t')-K_1(s,t))^2(\beta_0+\beta_1 r_{1,s}+\beta_2 \sqrt{r_{2,s}})^2 ds \right|^{\frac{p}{2}} \right]\\
    &+\mathbb{E}\left[\left|\int_{-\Delta}^0 (K_2(s,t')-K_2(s,t))(\beta_0+\beta_1 r_{1,s}+\beta_2 \sqrt{r_{2,s}})^2 ds \right|^p \right] \Biggr) \\
    &\le C |t'-t|^{\gamma p }
  \end{aligned}
\end{equation*}
where $C$ is a constant that can change from line to line. We deduce that the solution to Equation (\ref{eq:local_solution}) is locally $\gamma^*$-Hölder continuous for any $\gamma^*\in \left(0,\min\left(\frac{1}{2\alpha_1^*},\frac{1}{\alpha_2^*},\gamma\right) \right)$ by Theorem \ref{thm:continuity}. \\

We now show that if $\tau$ is a stopping time, then the uniqueness of the solution to (\ref{eq:local_solution}) holds on $[0,\tau)$. Let $R^{(n)}$ and $\tilde{R}^{(n)}$ be two solutions of (\ref{eq:local_solution}) on $[0,\tau)$ and set $Z^{(n)}_t = R^{(n)}_t-\tilde{R}^{(n)}_t$. We have:
\begin{equation*}
  \begin{aligned}
    \mathbb{E}\left[\|Z_t^{(n)} \mathbb{1}_{\{t <\tau \}}\|^2 \right] \le 2\Biggl(&\mathbb{E}\left[\left\|\int_0^{t\wedge \tau}K_2(s,t)\left(b_n\left(R_s^{(n)}\right)-b_n\left(\tilde{R}_s^{(n)}\right)\right)ds \right\|^2 \right]\\
    &+\mathbb{E}\left[\left\|\int_0^{t\wedge \tau} K_1(s,t)\left(\gamma_n\left(R_s^{(n)}\right)-\gamma_n\left(\tilde{R}_s^{(n)}\right)\right)dW_s \right\|^2 \right]\Biggr).
  \end{aligned}
\end{equation*}
By the Cauchy-Schwarz inequality, the Lipschitz property of $b_n$, assumption \ref{hyp:1.1} and Fubini's theorem for the first term, we obtain:
\begin{equation*}
  \begin{aligned}
  &\mathbb{E}\left[\left\|\int_0^{t\wedge \tau}K_2(s,t)\left(b_n\left(R_s^{(n)}\right)-b_n\left(\tilde{R}_s^{(n)}\right)\right)ds \right\|^2 \right] \\
  & \le C \int_0^{t} K_2(s,t) ds \mathbb{E}\left[\int_0^{t\wedge \tau} K_2(s,t)\left\|b_n\left(R_s^{(n)}\right)-b_n\left(\tilde{R}_s^{(n)}\right)\right\|^2 ds\right] \\
   &\le C \int_0^t K_2(s,t) \mathbb{E}\left[\left\|Z_s^{(n)}\mathbb{1}_{\{s < \tau\}} \right\|^2 \right]ds. 
  \end{aligned}
\end{equation*}
Similarly, by Doob's inequality for the martingale $\left(\int_0^r K_1(s,t)\left(\gamma_n\left(R_s^{(n)}\right)-\gamma_n\left(\tilde{R}_s^{(n)}\right)\right)dW_s   \right)_{r\in [0,t]}$, the Lipschitz property of $\gamma_n$ and Fubini's theorem, we get:
\begin{equation*}
  \begin{aligned}
    &\mathbb{E}\left[\left\|\int_0^{t\wedge \tau} K_1(s,t)\left(\gamma_n\left(R_s^{(n)}\right)-\gamma_n\left(\tilde{R}_s^{(n)}\right)\right)dW_s \right\|^2 \right] \\  
    &\le C \int_0^t K_1(s,t)^2   \mathbb{E}\left[\left\|Z_s^{(n)}\mathbb{1}_{\{s < \tau\}} \right\|^2 \right] ds .                              
  \end{aligned}
\end{equation*}
Thus, 
\begin{equation*}
  \mathbb{E}\left[\left\|Z_t^{(n)} \mathbb{1}_{\{t <\tau \}}\right\|^2 \right] \le C \int_0^t \left(K_2(s,t)+K_1(s,t)^2\right)\mathbb{E}\left[\left\|Z_s^{(n)}\mathbb{1}_{\{s<\tau \}}\right\|^2 \right] ds. 
\end{equation*}
By assumption \ref{hyp:1.2}, the generalised Grönwall lemma of \cite{zhang2010stochastic} (Lemma 2.2) can be applied which gives for $dt$ a.a. $t\in[0,T]$:

\begin{equation*}
  \mathbb{E}\left[\|Z_t^{(n)} \mathbb{1}_{\{t <\tau \}}\|^2 \right]  = 0
\end{equation*}
which, by the continuity of $R^{(n)}$ and $\tilde{R}^{(n)}$, implies that $R^{(n)}|_{[0,\tau)}=\tilde{R}^{(n)}|_{[0,\tau)}$. \\

At this stage, we can construct a solution $R$ to Equation (\ref{eq:r_dynamics}) up to a positive stopping time $\tau$. For $n\in{\mathbb N}^*$, we define the stopping time $\tau_n$ by
\begin{equation*}
  \tau_n:= \tau_n^1 \wedge \tau_n^2 
\end{equation*}
where: 
\begin{equation*}
  \begin{aligned}
    \tau_n^1&= \inf\left\{t\ge 0: \left|\beta_0+\beta_1R^{(n)}_{1,t}+\beta_2\sqrt{R^{(n)}_{2,t}}\right| \ge n \right\}, \\
    \tau_n^2&= \inf\left\{t\ge 0: R^{(n)}_{2,t}\le 1/n \right\}.
  \end{aligned}
\end{equation*}
Since $R^{(n)}$ and $R^{(n+1)}$ solve the same equation with coefficients $(b_{n+1},\gamma_{n+1})$ on $[0,\tau_n)$, we deduce from the above claim that $R^{(n)}|_{[0,\tau_n)}=R^{(n+1)}|_{[0,\tau_n)}$. Hence, $\tau_n \le \tau_{n+1}$ a.s. and, setting $\tau=\lim_{n\to +\infty}\tau_n$ and $\tau_0=0$, the process $\left(R_t:=\sum_{n\in{\mathbb N}^*}R^{(n)}_t\mathbb{1}_{\{\tau_{n-1}\le t<\tau_n\}}\right)_{t\in [0,\tau)}$ is a solution of Equation \eqref{eq:r_dynamics} up to $\tau$.
 By continuity of both processes $R^{(n)}$ and $R^{(n+1)}$, 
$R^{(n)}_{\tau_n} = R^{(n+1)}_{\tau_n}$
when $\tau_n$ is finite so that $R$ is locally $\gamma^*$-Hölder continuous on $[0,\tau)$ for any $\gamma^*\in \left(0,\min\left(\frac{1}{2\alpha_1^*},\frac{1}{\alpha_2^*},\gamma\right)\right)$. 

At this stage, we would like to show that $\mathbb{P}(\tau < \infty)=0$, i.e. that the constructed solution is global. This relies on the following lemma, the proof of which is postponed after Step 3: 
\begin{lemma}\label{lem:strong_bound}
  Under the assumptions \ref{hyp:1.1}-\ref{hyp:1.5}, we have for all $T>0$ and $p\ge 0$ :
  \begin{equation*}
    \sup_{n \in \mathbb{N}^*} \mathbb{E}\left[ \sup_{t\in [0,\hat{\tau}_n\wedge T]} (R_{1,t}^{2p}+R_{2,t}^{p}) \right] \le C_{T,p}
  \end{equation*}
  where $C_{T,p}$ is a constant depending only on $T$ and $p$ while $\hat{\tau}_n:= \hat{\tau}_n^1 \wedge \hat{\tau}_n^2$ with: 
  \begin{equation*}
    \begin{aligned}
      \hat{\tau}_n^1 &= \inf\left\{t\ge 0: \left|\beta_0+\beta_1R_{1,t}+\beta_2\sqrt{R_{2,t}} \right| \ge n \right\}, \\
      \hat{\tau}_n^2 &= \inf\left\{t\ge 0: R_{2,t} \le \frac{1}{n} \right\},
    \end{aligned}
  \end{equation*}
  $R$ being the solution constructed in Step 1. 
\end{lemma}

\paragraph{Step 2 of the proof of Theorem \ref{thm:main_result_1}: proof of $\mathbb{P}(\tau < \infty)=0$.} 

Let $T>0$, we have by Markov's inequality and Lemma \ref{lem:strong_bound}:
  \begin{equation*}
    \begin{aligned}
      \mathbb{P}(\hat{\tau}_n \le T) &\le \mathbb{P}\left(\sup_{t\in[0,T\wedge \hat{\tau}_n]} \left|\beta_0+\beta_1R_{1,t}+\beta_2\sqrt{R_{2,t}} \right| \ge n \right) + \mathbb{P}\left(\inf_{t\in[0,T\wedge \hat{\tau}_n]}R_{2,t}\le \frac{1}{n} \right) \\
      &\le \frac{1}{n}\mathbb{E}\left[\sup_{t\in[0,T\wedge \hat{\tau}_n]} \left|\beta_0+\beta_1R_{1,t}+\beta_2\sqrt{R_{2,t}} \right| \right]  +  \mathbb{P}\left(\inf_{t\in[0,T\wedge \hat{\tau}_n]}R_{2,t}\le \frac{1}{n} \right) \\
      &\le \frac{C_{T,\frac 12}}{n}+ \mathbb{P}\left(\inf_{t\in[0,T\wedge \hat{\tau}_n]}R_{2,t}\le \frac{1}{n} \right).
    \end{aligned}
  \end{equation*}
  Besides, by assumption \ref{hyp:1.6}, we have the following deterministic lower bound which does not depend on $t$:
  \begin{equation*}
    \begin{aligned}
      R_{2,t} &= \int_{-\Delta}^0K_2(s,t)\sigma_s^2 ds + \int_0^t K_2(s,t) \sigma_s^2 ds \\
      &\ge \inf_{t'\in [0,T]} \underbrace{\int_{-\Delta}^0  K_2(s,t')\sigma_s^2 ds}_{g_2(t')} >0.
    \end{aligned}
  \end{equation*}
  We deduce that $\mathbb{P}\left(\inf_{t\in[0,T\wedge \hat{\tau}_n]}R_{2,t}\le \frac{1}{n} \right) = 0$ for $n$ large enough and consequently that $\mathbb{P}(\hat{\tau}_n \le T) \xrightarrow[n\to+\infty]{} 0$ for all $T>0$. Hence, $\mathbb{P}(\tau < \infty)=0$ since $\hat{\tau}_n \nearrow \tau$.\\

  \paragraph{Step 3: solution's uniqueness.} Let $R$ be the global continuous solution to Equation (\ref{eq:r_dynamics}) that we have constructed and let $\tilde{R}$ be another global continuous solution. Let us set $\nu_n = \nu_n^1 \wedge \nu_n^2$, $\tilde{\nu}_n =\tilde{\nu}_n^1 \wedge \tilde{\nu}_n^2$ and $S_n = \nu_n \wedge \tilde{\nu}_n$ where:
  \begin{equation*}
    \begin{aligned}
      \nu_n^1 &= \inf\left\{t\ge 0: \left|\beta_0+\beta_1R_{1,t}+\beta_2\sqrt{R_{2,t}} \right| \ge n \right\}, \\
      \nu_n^2 &= \inf\left\{t\ge 0: R_{2,t} \le \frac{1}{n} \right\} 
    \end{aligned}
  \end{equation*}
  and where $\tilde{\nu}_n^1$, $\tilde{\nu}_n^2$ are similarly defined for $\tilde{R}$ instead of $R$. Let $0\le t \le T$, we have:
  \begin{equation*}
    \begin{aligned}
      &\mathbb{E}\left[\|R_t-\tilde{R}_t \|^2 \mathbb{1}_{\{t \le S_n \}}\right] \le C \Biggl( \underbrace{\mathbb{E}\left[\left| \int_0^{t} K_1(s,t)\left(\beta_1(R_{1,s}-\tilde{R}_{1,s})+\beta_2(\sqrt{R_{2,s}}-\sqrt{\tilde{R}_{2,s}}) \right)\mathbb{1}_{\{s \le S_n \}} dW_s  \right|^2\right]}_{I_1}  \\
      &+ \underbrace{\mathbb{E}\left[\left|\int_0^{t} K_2(s,t)\left(\left( \beta_0+\beta_1R_{1,s}+\beta_2\sqrt{R_{2,s}}\right)^2- \left(\beta_0+\beta_1\tilde{R}_{1,s}+\beta_2\sqrt{\tilde{R}_{2,s}} \right)^2  \right) \mathbb{1}_{\{s \le S_n \}}ds \right|^2\right]}_{I_2}\Biggr) 
    \end{aligned}
  \end{equation*}
  where $C$ is a constant. In the sequel, this constant can change from line to line and may depend on $n$ but not on $t$. By Burkholder-Davis-Gundy's inequality (abbreviated BDG in the sequel) applied to the local martingale $\left(\int_0^r  K_1(s,t)\left(\beta_1(R_{1,s}-\tilde{R}_{1,s})+\beta_2(\sqrt{R_{2,s}}-\sqrt{\tilde{R}_{2,s}}) \right)\mathbb{1}_{\{s \le S_n \}} dW_s\right)_{r\in[0,t]}$ and the fact that $\sqrt{R_{2,s}}$ and $\sqrt{\tilde{R}_{2,s}}$ are bounded from below by $1/\sqrt{n}$ for $s\le S_n$, we have:
  \begin{equation*}
    \begin{aligned}
    I_1 &\le C \mathbb{E}\left[\int_0^t K_1(s,t)^2\left(\beta_1(R_{1,s}-\tilde{R}_{1,s})+\beta_2(\sqrt{R_{2,s}}-\sqrt{\tilde{R}_{2,s}}) \right)^2  \mathbb{1}_{\{s \le S_n \}}ds\right] \\
  & \le C \mathbb{E}\left[\int_0^t 2 K_1(s,t)^2\left(\beta_1^2(R_{1,s}-\tilde{R}_{1,s})^2+\frac{\beta_2^2n}{4} (R_{2,s}-\tilde{R}_{2,s})^2 \right)  \mathbb{1}_{\{s \le S_n \}}ds\right]\\
  &\le C \int_0^t K_1(s,t)^2 \mathbb{E}\left[\|R_s - \tilde{R}_s\|^2 \mathbb{1}_{\{s\le S_n\}}\right] ds. 
    \end{aligned}
  \end{equation*}
  Moreover, using the fact that $2\beta_0+\beta_1(R_{1,s}+\tilde{R}_{1,s})+\beta_2\left(\sqrt{R_{2,s}}+\sqrt{\tilde{R}_{2,s}}\right)\le 2n $ and $\sqrt{R_{2,s}}+\sqrt{\tilde{R}_{2,s}}\ge 2/\sqrt{n}$ for $s\le S_n$, by the Cauchy-Schwarz inequality and by assumption \ref{hyp:1.1}, we get:
  \begin{equation*}
    \begin{aligned}
      I_2 &\le 4n^2 \mathbb{E}\left[\left|\int_0^t K_2(s,t)\left(\beta_1\left(R_{1,s}-\tilde{R}_{1,s}\right)+\beta_2\left(\sqrt{R_{2,s}}-\sqrt{\tilde{R}_{2,s}}\right)\right)\mathbb{1}_{\{s\le S_n\}}ds \right|^2 \right]\\
      &\le 4n^2 \int_0^t K_2(s,t)ds \times \mathbb{E}\left[\int_0^t K_2(s,t)  \left(\beta_1\left(R_{1,s}-\tilde{R}_{1,s}\right)+\beta_2\left(\sqrt{R_{2,s}}-\sqrt{\tilde{R}_{2,s}}\right)\right)^2\mathbb{1}_{\{s\le S_n\}}ds \right]  \\
      &\le C \int_0^t K_2(s,t)\mathbb{E}\left[\|R_s - \tilde{R}_s\|^2 \mathbb{1}_{\{s\le S_n\}}\right] ds.
    \end{aligned}
  \end{equation*}
  We deduce that:
  \begin{equation*}
    \mathbb{E}\left[ \|R_t-\tilde{R}_t \|^2 \mathbb{1}_{\{t \le S_n \}}\right] \le C \int_0^t \left(K_1(s,t)^2+ K_2(s,t)\right) \mathbb{E}\left[\|R_s-\tilde{R}_s \|^2 \mathbb{1}_{\{s\le S_n \}} \right] ds 
  \end{equation*}
  Using the generalized Grönwall's lemma in \cite{zhang2010stochastic} (Lemma 2.2) and the continuity of $R$ and $\tilde{R}$, we obtain that $\mathbb{E}\left[ \|R_t-\tilde{R}_t \|^2 \mathbb{1}_{\{t \le S_n \}}\right] = 0$. Thus, almost surely, 
  \begin{equation*}
    (R_t-\tilde{R}_t) \mathbb{1}_{\{t\le S_n\}} = 0. 
  \end{equation*}
  When $S_n$ is finite, since $R$ and $\tilde{R}$ coincide on $[0,S_n]$ and either 
  \begin{equation*}
    \left|\beta_0+\beta_1R_{1,S_n}+\beta_2R_{2,S_n} \right| =  \left|\beta_0+\beta_1\tilde{R}_{1,S_n}+\beta_2\tilde{R}_{2,S_n} \right| \ge n 
  \end{equation*}
  or 
  \begin{equation*}
    R_{2,S_n} = \tilde{R}_{2,S_n} \le \frac{1}{n},
  \end{equation*}
then we have necessarily $S_n=\nu_n = \tilde{\nu}_n$ by definition of $\nu_n$ and $\tilde{\nu}_n$. According to the second step of the proof of Theorem \ref{thm:main_result_1}, $\nu_n \xrightarrow[n\to +\infty]{a.s.}+\infty$. Hence, $S_n \xrightarrow[n\to +\infty]{a.s.}+\infty$ and $R_t=\tilde{R}_t$ for all $t\ge 0$. \\

We conclude this section with the proof of Lemma \ref{lem:strong_bound} which relies on the two next lemmas. 

\begin{lemma}\label{lem:moments}
  Under the assumptions \ref{hyp:1.1}-\ref{hyp:1.5}, we have for all $T>0$, $p\ge 0$ and $n\in \mathbb{N}^*$:
  \begin{equation*}
    \sup_{t\in [0,T]} \mathbb{E}\left[ (|R_{1,t}|^{2p}+R_{2,t}^{p})\mathbb{1}_{\{t \le \hat{\tau}_n \}} \right] \le C_{T,p}
  \end{equation*}
  where $R$ and $\hat{\tau}_n$ are defined respectively in Step 1 and 2, and $C_{T,p}$ is a constant depending only on $T$ and $p$. 
\end{lemma}
\begin{proof}
  Let $T>0$, $t\in [0,T]$, $n\in \mathbb{N}^*$ and $p\ge \max(\alpha_1^*,\alpha_2^*)$ where $\alpha_i^*:=\frac{\alpha_i}{\alpha_i-1}$ and $\alpha_1$, $\alpha_2$ are defined in assumption \ref{hyp:1.4}. On the one hand, we have for $R_1$:
  \begin{equation*}
      \mathbb{E}\left[|R_{1,t}|^{2p} \mathbb{1}_{\{t\le \hat{\tau}_n\}} \right] \le 2^{p-1}\left(\mathbb{E}[|g_1(t)|^{2p}]+\mathbb{E}\left[\left|\int_0^{t\wedge \hat{\tau}_n}K_1(s,t)\left(\beta_0+\beta_1R_{1,s}+\beta_2\sqrt{R_{2,s}} \right)dW_s \right|^{2p} \right] \right).
  \end{equation*}
Applying BDG's inequality and twice Hölder's inequality, we get for the second term:
  \begin{equation*}
    \begin{aligned}
      &\mathbb{E}\left[\left|\int_0^{t\wedge \hat{\tau}_n}K_1(s,t)\left(\beta_0+\beta_1R_{1,s}+\beta_2\sqrt{R_{2,s}} \right)dW_s \right|^{2p} \right] \\
      &\le \mathbb{E}\left[\left(\int_0^{t\wedge \hat{\tau}_n} K_1(s,t)^2\left(\beta_0+\beta_1R_{1,s}+\beta_2\sqrt{R_{2,s}}\right)^2ds \right)^{p}\right] \\
      &\le  \left(\int_0^t K_1(s,t)^{2\alpha_1}\right)^{\frac{p}{\alpha_1}}\mathbb{E}\left[\left(\int_0^{t\wedge \hat{\tau}_n}\left|\beta_0+\beta_1R_{1,s}+\beta_2\sqrt{R_{2,s}}\right|^{2\alpha_1^*}ds\right)^{\frac{p}{\alpha_1^*}}\right] \\
      &\le  \left(\int_0^t K_1(s,t)^{2\alpha_1}\right)^{\frac{p}{\alpha_1}} T^{\frac{p}{\alpha_1^*}-1}\int_0^{t}\mathbb{E}\left[\left|\beta_0+\beta_1R_{1,s}\mathbb{1}_{\{s\le \hat{\tau}_n\}}+\beta_2\sqrt{R_{2,s}\mathbb{1}_{\{s\le \hat{\tau}_n\}}}\right|^{2p}\right] ds \\
      &\le C_{p,T}\left(1+\int_0^t \mathbb{E}[(|R_{1,s}|^{2p}+R_{2,s}^p)\mathbb{1}_{\{s\le \hat{\tau}_n\}}]ds \right)
    \end{aligned}
  \end{equation*}
where we used assumption \ref{hyp:1.4} for the last line. On the other hand, we have for $R_2$:
\begin{equation*}
  \mathbb{E}\left[R_{2,t}^p\mathbb{1}_{\{t\le \hat{\tau}_n \}} \right] \le 2^{p-1}\left(\mathbb{E}[g_2(t)^p]+\mathbb{E}\left[\left(\int_0^{t\wedge \hat{\tau}_n} K_2(s,t)\left(\beta_0+\beta_1R_{1,s}+\beta_2\sqrt{R_{2,s}} \right)^2ds  \right)^p \right] \right).
\end{equation*}
Applying again twice Hölder's inequality and using assumption \ref{hyp:1.4} for the second term, we obtain similarly:
\begin{equation*}
    \mathbb{E}\left[\left(\int_0^{t\wedge \hat{\tau}_n} K_2(s,t)\left(\beta_0+\beta_1R_{1,s}+\beta_2\sqrt{R_{2,s}} \right)^2ds  \right)^p \right]\le C_{T,p}\left(1+\int_0^t \mathbb{E}[(|R_{1,s}|^{2p}+R_{2,s}^p)\mathbb{1}_{\{s\le \hat{\tau}_n\}}]ds \right).
\end{equation*}
Bringing together the four previous equations and using the fact that $\sup_{t\in [0,T]}\mathbb{E}[\|g(t)\|^p] < \infty$ by conditions \ref{hyp:1.1} and \ref{hyp:1.3}, we deduce that:
\begin{equation*}
  \mathbb{E}\left[(|R_{1,t}|^{2p}+R_{2,t}^p)\mathbb{1}_{\{t\le \hat{\tau}_n \}} \right] \le C_{T,p}\left(1+\int_0^t \mathbb{E}[(|R_{1,s}|^{2p}+R_{2,s}^p)\mathbb{1}_{\{s\le \hat{\tau}_n\}}]ds \right).
\end{equation*}
Since 
\begin{equation*}
  \mathbb{E}[(|R_{1,s}|^{2p}+R_{2,s}^p)\mathbb{1}_{\{s\le \hat{\tau}_n\}}]=\mathbb{E}\left[\left(\left|R_{1,s}^{(n)}\right|^{2p}+\left|R_{2,s}^{(n)}\right|^{p} \right)\mathbb{1}_{\{s\le \hat{\tau}_n\}} \right] 
  \le  \sup_{s\in [0,T]} \mathbb{E}\left[ \left\|R_t^{(n)}\right\|^{2p} + \left\|R_t^{(n)}\right\|^{p} \right] 
\end{equation*}
which is finite by (\ref{eq:boundedness_Rn}), we can apply Grönwall's lemma to obtain:
\begin{equation*}
  \mathbb{E}\left[(|R_{1,t}|^{2p}+R_{2,t}^p)\mathbb{1}_{\{t\le \hat{\tau}_n \}} \right] \le C_{T,p}. 
\end{equation*}
\end{proof}

The following lemma is a consequence of the Garsia-Rodemich-Rumsey inequality (\citeauthor{garsia1970real}, \citeyear{garsia1970real}) and is inspired from Appendix A.3 in \cite{nualart2006}. 

\begin{lemma}\label{lem:grr}
  Let $(X_t)_{t\ge 0}$ be a real-valued process and $\tau$ be a stopping time bounded by a constant $T>0$. If there exists $p>0$, $C>0$ and $\varepsilon>0$ such that for all $s,t\ge 0$:
  \begin{equation} \label{eq:ineq_grr}
    \mathbb{E}\left[|X_t-X_s|^p \mathbb{1}_{\{s,t \in [0,\tau] \}} \right] \le C |t-s|^{1+\varepsilon},
  \end{equation}
  and $\mathbb{E}[|X_{t^*}|^p] <\infty$ for some $t^*$ such that $\mathbb{P}(t^*\le \tau)=1$, then:
  \begin{equation*}
    \mathbb{E}\left[\sup_{t\in [0,\tau]} |X_t|^p \right] < \infty. 
  \end{equation*}
\end{lemma}
\begin{proof}
  Define $\psi(x)=|x|^p$ and $q(x)= |x|^{\frac{\gamma}{p}}$ for $0 < \gamma < (2+\varepsilon)\wedge 2(p+1)$. We have:
  \begin{equation*}
    \begin{aligned}
      B&:= \int_{[0,\tau]^2} \psi\left(\frac{X_t-X_s}{q(t-s)} \right)ds dt \\
      &= \int_{[0,T]^2} \frac{|X_t-X_s|^p}{|t-s|^{\gamma}}\mathbb{1}_{\{ s,t \in [0,\tau]\}} ds dt.
    \end{aligned}
  \end{equation*}
  Taking the expectation on both sides, we get by Fubini's theorem and inequality (\ref{eq:ineq_grr}):
  \begin{equation*}
    \mathbb{E}[B] \le \int_{[0,T]^2} |t-s|^{1+\varepsilon-\gamma}ds dt. 
  \end{equation*}
  The right hand side is finite since $0 < \gamma < 2+\varepsilon$, therefore $\mathbb{E}[B]$ is finite as well and we deduce that $B$ is almost surely finite. By the Garsia-Rodemich-Rumsey inequality (\citeauthor{garsia1970real}, \citeyear{garsia1970real}), we obtain that almost surely for $s,t \in [0,\tau]$:
  \begin{equation*}
    |X_t-X_s| \le 8 \int_0^{|t-s|} \psi^{-1}\left(\frac{4B}{u^2} \right) q(du) 
  \end{equation*}
  where $\psi^{-1}(u) := \sup \{v: \psi(v)\le u \} = |u|^{1/p}$. Noting that $q(du)= \frac{\gamma}{p}u^{\frac{\gamma}{p}-1}du$ for $u >0$, we get:
  \begin{equation*}
    \begin{aligned}
      |X_t-X_s| &\le C_{p,\gamma} B^{\frac{1}{p}} \int_0^{|t-s|} u^{\frac{\gamma-2}{p}-1} du  \\ 
      &\le C_{p,\gamma,T} B^{\frac{1}{p}}
    \end{aligned}
  \end{equation*}
  where $C_{p,\gamma,T}$ is a deterministic constant depending only on $p$, $\gamma$ and $T$. We deduce that:
  \begin{equation*}
   \sup_{t\in [0,\tau]} |X_t|^p \le |X_{t^*}|^p+ C_{p,\gamma,T}B. 
  \end{equation*}
  Hence $\mathbb{E}\left[\sup_{t\in [0,\tau]} |X_t|^p \right] < \infty$ since $\mathbb{E}[|X_{t^*}|^p] <\infty$ and $\mathbb{E}[B]<\infty$. 
\end{proof}
We can now prove Lemma \ref{lem:strong_bound}. 

\begin{proof}[Proof of Lemma \ref{lem:strong_bound}]
The idea of the proof is to show that:
\begin{equation*}
  \mathbb{E}[|R_{1,t'}-R_{1,t}|^{2p}\mathbb{1}_{\{t,t'\in[0,\hat{\tau}_n \wedge T]\}}] + \mathbb{E}[|R_{2,t'}-R_{2,t}|^{p}\mathbb{1}_{\{t,t'\in[0,\hat{\tau}_n \wedge T]\}}] \le C |t'-t|^{1+\varepsilon}
\end{equation*}
for $0\le t < t' \le T$ with $C,\varepsilon >0$ so that we can conclude by applying Lemma \ref{lem:grr} since $\mathbb{E}[|R_{1,0}|^{2p}+R_{2,0}^p]=\mathbb{E}[|g_1(0)|^{2p}+g_2(0)^p]<\infty$ (see Step 1). We start with the first term. Let $0\le t < t'\le T$ and $p > \min \left(\frac{1}{\gamma},\alpha_1^*,\alpha_2^* \right)$ where $\alpha_i^* := \frac{\alpha_i}{\alpha_i-1}$ and $\alpha_1$, $\alpha_2$, $\gamma$ are defined in assumptions \ref{hyp:1.4} and \ref{hyp:1.5}. We have:
\begin{equation*}
  \begin{aligned}
    \mathbb{E}[|R_{1,t'}-R_{1,t}|^{2p}\mathbb{1}_{\{t,t'\in[0,\hat{\tau}_n \wedge T]\}}] \le & C\Biggl(\underbrace{\mathbb{E}[|g_1(t')-g_1(t)|^{2p}]}_{I_1(t,t')} + \underbrace{\mathbb{E}\left[\left|\int_0^{t} (K_1(s,t')-K_1(s,t))\gamma(R_s)\mathbb{1}_{\{s\le \hat{\tau}_n\}}dW_s \right|^{2p}\right]}_{I_2(t,t')} \\
    &+ \underbrace{\mathbb{E}\left[\left|\int_t^{t'} K_1(s,t')\gamma(R_s)\mathbb{1}_{\{s\le \hat{\tau}_n\}}dW_s  \right|^{2p} \right]}_{I_3(t,t')} \Biggr) 
  \end{aligned}
\end{equation*}
where $C$ is a constant that can change from line to line in the sequel but which does not depend on $n$. We already showed in the Step 1 that:
\begin{equation*}
  \begin{aligned}
  I_1(t,t')&\le C \mathbb{E}\left[\left|\int_{-\Delta}^0 (K_1(s,t')-K_1(s,t))^2 (\beta_0+\beta_1r_{1,s}+\beta_2\sqrt{r_{2,s}})^2 ds  \right|^p \right] \\
  &\le C |t'-t|^{2\gamma p}. 
  \end{aligned}
\end{equation*}
By BDG's inequality, Hölder's inequality, Lemma \ref{lem:moments} and assumption \ref{hyp:1.5}, we get:
\begin{equation*}
  \begin{aligned}
  I_2(t,t')&\le C \mathbb{E}\left[\left|\int_{0}^t (K_1(s,t')-K_1(s,t))^2 (\beta_0+\beta_1R_{1,s}+\beta_2\sqrt{R_{2,s}})^2\mathbb{1}_{\{s\le \hat{\tau}_n\}} ds  \right|^p \right] \\
  &\le C\left(\int_0^t (K_1(s,t')-K_1(s,t))^2ds \right)^{p-1}\\
  &\times \int_0^t (K_1(s,t')-K_1(s,t))^2\mathbb{E}\left[\left|\beta_0+\beta_1R_{1,s}+\beta_2\sqrt{R_{2,s}}\right|^{2p}\mathbb{1}_{\{s\le \hat{\tau}_n\}} \right] ds\\ 
  &\le C \left(\int_0^t (K_1(s,t')-K_1(s,t))^2ds \right)^{p}\\
  &\le C |t'-t|^{2\gamma p }.
  \end{aligned}
\end{equation*}
Finally, by BDG's inequality, twice Hölder's inequality, assumption \ref{hyp:1.4} and Lemma \ref{lem:moments}, we have:
\begin{equation*}
  \begin{aligned}
    I_3(t,t') &\le C\mathbb{E}\left[\left|\int_t^{t'}K_1(s,t')^2\left(\beta_0+\beta_1R_{1,s}+\beta_2\sqrt{R_{2,s}}\right)^2 \mathbb{1}_{\{s\le \hat{\tau}_n\}}ds \right|^p\right]\\
    &\le C \left(\int_t^{t'} K_1(s,t')^{2\alpha_1} ds\right)^{\frac{p}{\alpha_1}} \mathbb{E}\left[\left| \int_t^{t'}\left(\beta_0+\beta_1R_{1,s}+\beta_2\sqrt{R_{2,s}}\right)^{2\alpha_1^*} \mathbb{1}_{\{s\le \hat{\tau}_n\}} ds \right|^{\frac{p}{\alpha_1^*}} \right] \\
    &\le C |t'-t|^{\frac{p}{\alpha_1^*}-1} \int_t^{t'}\mathbb{E}\left[ \left(\beta_0+\beta_1R_{1,s}+\beta_2\sqrt{R_{2,s}}\right)^{2p} \mathbb{1}_{\{s\le \hat{\tau}_n\}} \right] ds \\
    &\le C |t'-t|^{\frac{p}{\alpha_1^*}}.
  \end{aligned}
\end{equation*}
Combining the four previous equations, we obtain:
\begin{equation*}
  \mathbb{E}[|R_{1,t'}-R_{1,t}|^{2p}\mathbb{1}_{\{t,t'\in[0,\hat{\tau}_n \wedge T]\}}] \le C|t'-t|^{\xi_1}
\end{equation*}
where $\xi_1 = 2\gamma p \wedge \frac{p}{\alpha_1^*}  >1$. Using similar arguments, we can show that:
\begin{equation*}
  \mathbb{E}[|R_{2,t'}-R_{2,t}|^{p}\mathbb{1}_{\{t,t'\in[0,\hat{\tau}_n \wedge T]\}}] \le C|t'-t|^{\xi_2}
\end{equation*} 
where $\xi_2 = \gamma p \wedge  \frac{p}{\alpha_2^*}>1$. 
\end{proof}


\section{Proof of Theorem \ref{thm:main_result_2}}
In this section, $R=(R_1,R_2)$ is the local continuous solution that has been constructed in Step 1 of the proof of Theorem \ref{thm:main_result_1} up to a stopping time $\tau$ (note that this construction requires only assumption \ref{hyp:1.1}-\ref{hyp:1.5}). We will first prove that the volatility $\sigma$ is bounded from below by a positive stochastic process, and then we will show that the condition \ref{hyp:1.6} ensuring that the solution is global is implied by the additional assumptions we have made to get the positivity of the volatility. We first introduce two lemmas. 

\begin{lemma}\label{lem:differential_K}
Under assumptions \ref{hyp:1.1}-\ref{hyp:1.5} and \ref{hyp:2.1}, $R_1$ and $R_2$ are Itô processes on $[0,\tau)$ with decompositions:
 \begin{equation*}
  \begin{aligned}
   dR_{1,t} &= \left(\int_{-\Delta}^0 \partial_t K_1(u,t)\sigma_u dB_{-u}+\int_{0}^t \partial_t K_1(u,t)\sigma_u dW_u \right)dt + K_1(t,t)\sigma_t dW_t, \\ 
   dR_{2,t} &= \left(\int_{-\Delta}^t\partial_t K_2(u,t)\sigma_u^2du  \right)dt  + K_2(t,t)\sigma_t^2 dt 
  \end{aligned}
 \end{equation*}
 where we recall that $\sigma_t:= \beta_0+\beta_1R_{1,t}+\beta_2\sqrt{R_{2,t}}$. 
 \end{lemma}
   \begin{proof}
 We have by definition:
   \begin{equation*}
     R_{1,t} = \int_{-\Delta}^0 K_1(u,t)\sigma_udB_{-u} + \int_0^t K_1(u,t)\sigma_udW_u,
   \end{equation*}
   By the fundamental theorem of calculus, we can write $K_1(u,t) = K_1(u,s) + \int_s^t \partial_v K_1(u,v)dv$ for all $0\le u \le s\le t$. Thus, 
   \begin{equation}\label{eq:proof_r1}
       R_{1,t} = \int_{-\Delta}^0 \left( K_1(u,0) +\int_{0}^t \partial_v K_1(u,v)dv \right)\sigma_u dB_{-u} + \int_0^t \left( K_1(u,u)+\int_u^t \partial_v K_1(u,v)dv \right) \sigma_u dW_u. \\ 
   \end{equation}
Using assumption \ref{hyp:1.3} and \ref{hyp:2.1}, we have almost surely:
   \begin{equation*}
     \begin{aligned}
      \int_0^t \left(\int_{-\Delta}^0  \left|\partial_v K_1(u,v) \right|^2 \sigma_u^2 du \right)^{1/2} dv &\le \sup_{s\in(-\Delta,0]}\sigma_s \times  \int_0^t \left(\int_{-\Delta}^0 \left|\partial_v K_1(u,v) \right|^2 dv \right)^{1/2} du <\infty.
     \end{aligned}
   \end{equation*}
Similarly, using the continuity of $\sigma$ deduced from the continuity of $R$, we can show that:
   \begin{equation*}
     \int_0^t \left(\int_{0}^v \left|\partial_v K_1(u,v) \right|^2\sigma_u^2 du\right)^{1/2} dv  < \infty \ a.s. 
   \end{equation*}
   Hence, we can apply the stochastic Fubini's theorem (see \citeauthor{veraar2012}, \citeyear{veraar2012}) in Equation (\ref{eq:proof_r1}) to obtain:
   \begin{equation*}
     \begin{aligned}
     R_{1,t} &= r_{1,0}+ \int_0^t \left(\int_{-\Delta}^0\partial_v K_1(u,v) \sigma_u dB_{-u}\right) dv + \int_0^t K_1(u,u)\sigma_udW_u + \int_0^t \left(\int_0^v \partial_v K_1(u,v)\sigma_udW_u\right)dv \\
     &= r_{1,0}+ \int_0^t K_1(u,u)\sigma_udW_u + \int_0^t \left(\int_{-\Delta}^0 \partial_v K_1(u,v)\sigma_u dB_{-u}+\int_{0}^v \partial_v K_1(u,v)\sigma_u dW_u \right)dv.
     \end{aligned}
   \end{equation*}
  Using similar arguments, we also obtain:
  \begin{equation*}
    R_{2,t} = r_{2,0}+\int_0^t K_2(u,u)\sigma_u^2 du + \int_0^t \int_{-\Delta}^v \partial_v K_2(u,v) \sigma_u^2 du dv.
  \end{equation*}
   \end{proof}

The following lemma is an extension of the comparison result for solutions of SDEs of \customcite[Proposition~5.2.18]{karatzas1991}. 
\begin{lemma}\label{lem:comparison_result}
  Let $X$ and $Y$ be real-valued continuous Itô processes with decompositions:
  \begin{equation*}
    \begin{aligned}
      X_t &= X_0 + \int_0^t b_X\left(s,X_s\right) ds+\int_0^t \gamma(s,X_s) dW_s \\
      Y_t &= Y_0 + \int_0^t G_s ds+\int_0^t \gamma(s,Y_s) dW_s
    \end{aligned}
  \end{equation*}
Under the following assumptions:
  \begin{enumerate}[label=(\roman*)]
    \item $X_0\le Y_0$ a.s.;
    \item $b_X$ is a real-valued function on $\mathbb{R}_+ \times \mathbb{R}$ such that:
    \begin{equation*}
      |b_X(t,x)-b_X(t,y)| \le C|x-y| \quad \forall t\in \mathbb{R}_+, \forall x,y \in \mathbb{R}
    \end{equation*}
    for a positive constant $C$; \label{hyp:b_ineq}
    \item $G$ satisfies:
    \begin{equation*}
      b_X(t,Y_t) \le G_t  \quad \forall t\ge 0 \ a.s. ;
    \end{equation*}
    \item $\gamma$ is a real-valued function on $\mathbb{R}_+\times \mathbb{R}$ such that:
    \begin{equation*}
      |\gamma(t,x)-\gamma(t,y)|\le \ell(t) h(|x-y|) \quad \forall t\ge 0,\ \forall (x,y) \in \mathbb{R}^2
    \end{equation*}
    for $\ell:\mathbb{R}_+\rightarrow \mathbb{R}_+$ a locally square-integrable function and $h:\mathbb{R}_+\rightarrow \mathbb{R}_+$ a strictly increasing function with $h(0)=0$ and 
    \begin{equation*}
      \int_0^{\varepsilon} h^{-2}(u)du = +\infty, \quad \forall \varepsilon >0;
    \end{equation*}\label{hyp:sigma_reg}
  \end{enumerate}
we have $\mathbb{P}(X_t\le Y_t,\ \forall t\ge 0)=1$. 
\end{lemma}
\begin{proof}
  Using the assumptions on the function $h$, we can construct a strictly decreasing sequence $(a_n)_{n\in \mathbb{N}}$ with $a_0=1$, $\lim_{n\to +\infty} a_n = 0$ and $\int_{a_n}^{a_{n-1}} h^{-2}(u)du = n$ for all $n\ge 1$. Moreover, we can construct a sequence $(\rho_n)_{n\ge 1}$ of continuous functions on $\mathbb{R}$ with support in $(a_n,a_{n-1})$ such that:
  \begin{itemize}
    \item $0\le \rho_n(x) \le \frac{2}{nh^2(x)}$ for all $x>0$ and
    \item $\int_{a_n}^{a_{n-1}}\rho_n(x)dx = 1$  
  \end{itemize} 
because the upper bound in the first item integrates to 2 over $(a_n,a_{n-1})$. Let us now define the sequence $(\psi_n)_{n\ge 1}$ as:
\begin{equation*}
  \psi_n(x) = \int_0^{x}\int_0^y \rho_n(u)du dy \quad \text{for } x \in \mathbb{R}.
\end{equation*}
For all $n\ge 1$, $\psi_n$ is twice continuously differentiable ($\psi_n(x) = 0$ for $x\le a_n$) with $\psi_n'(x) = \int_0^{x} \rho_n(u)du \in [0,1]$. Besides, $\psi_n(x) \xrightarrow[n\to+\infty]{} x^+$ since $\psi_n(x)=0$ if $x\le 0$, $\psi_n(x) \le \int_0^{x} dy = x$ and for $x\ge a_{n-1}$,
\begin{equation*}
  \begin{aligned}
  \psi_n(x) &= \int_{a_n}^{x}\int_{a_n}^y \rho_n(u)du dy \\
  &= \underbrace{\int_{a_n}^{a_{n-1}}\int_{a_n}^y \rho_n(u)du dy}_{\ge 0} + \int_{a_{n-1}}^{x}\underbrace{\int_{a_n}^y \rho_n(u)du}_{=1} dy \\
  &\ge \int_{a_{n-1}}^{x} du = x -a_{n-1}.
  \end{aligned}
\end{equation*}
Using a similar decomposition, we can also show that $\psi_{n+1}(x)-\psi_n(x)\ge 0$ for all $x \in \mathbb{R}$, i.e. that the sequence $(\psi_n)_{n\ge 1}$ is non-decreasing. By a localization argument, we can assume that without loss of generality:
\begin{equation}\label{eq:martingale_condition}
  \mathbb{E}\left[\int_0^t |\gamma(s,X_s)|^2+|\gamma(s,Y_s)|^2ds \right] < \infty
\end{equation}
since $\int_0^t \gamma(s,X_s)^2 + \gamma(s,Y_s)^2ds$ is finite for any $t\ge 0$ almost surely. We have:
\begin{equation*}
  \begin{aligned}
  \Delta_t &:= X_t-Y_t  \\
  &= X_0-Y_0+ \int_0^t \left(b_X(s,X_s)-G_s\right)ds + \int_0^t \left(\gamma(s,X_s)-\gamma(s,Y_s) \right)dW_s.
  \end{aligned}
\end{equation*}
Thus, by Itô's formula, 
\begin{equation*}
  \begin{aligned}
  \psi_n(\Delta_t) &= \psi_n\left(X_0-Y_0\right)+\int_0^t \psi_n'(\Delta_s)\left(b_X(s,X_s)-G_s\right)ds  \\
  &+ \int_0^t\psi_n'(\Delta_s)\left(\gamma(s,X_s)-\gamma(s,Y_s) \right)dW_s + \frac{1}{2}\int_0^t \psi_n''(\Delta_s)\left(\gamma(s,X_s)-\gamma(s,Y_s) \right)^2ds.
  \end{aligned}
\end{equation*}
By assumption (\ref{eq:martingale_condition}) and $0\le \psi_n' \le 1$, the stochastic integral is a martingale so its expectation is 0. Besides, $\psi_n\left(X_0-Y_0\right)\le 0$ since $\psi_n$ is non-decreasing. Thus, taking the expectation on both sides in the above equation, we get:
\begin{equation*}
  \begin{aligned}
  \mathbb{E}\left[ \psi_n(\Delta_t)\right] &\le \mathbb{E}\left[ \int_0^t \psi_n'(\Delta_s)\left(b_X(s,X_s)-G_s\right)ds\right]\\
  & +  \frac{1}{2}\mathbb{E}\left[\int_0^t \psi_n''(\Delta_s)\left(\gamma(s,X_s)-\gamma(s,Y_s) \right)^2ds \right].  
  \end{aligned}
\end{equation*}
Using that $\psi_n''(x)=\rho_n(|x|)$, $\rho_n(x) \le \frac{2}{nh^2(x)}$ and assumption \ref{hyp:sigma_reg}, we obtain:
\begin{equation*}
  \frac{1}{2}\mathbb{E}\left[\int_0^t \psi_n''(\Delta_s)\left(\gamma(s,X_s)-\gamma(s,Y_s) \right)^2ds \right] \le \frac{1}{n} \int_0^t \ell(s)^2 ds. 
\end{equation*}
Finally, by assumption \ref{hyp:b_ineq}, remark that:
\begin{equation*}
  \begin{aligned}
  b_X(s,X_s)-G_s &=  b_X(s,X_s)- b_X(s,Y_s) + \underbrace{b_X(s,Y_s)-G_s}_{\le 0}\\
  &\le C|\Delta_s|.
  \end{aligned}
\end{equation*}
Since $\psi_n'(x)=0$ for $x\le 0$ and $\psi_n'\le 1$, we have therefore:
\begin{equation*}
  \psi_n'(\Delta_s) \left(b_X(s,X_s)-G_s  \right) \le C \Delta_s^+.   
\end{equation*}
Hence,
\begin{equation*}
  \mathbb{E}\left[ \psi_n(\Delta_t)\right] \le C\int_0^t \mathbb{E}[\Delta_s^+] ds +  \frac{1}{n} \int_0^t \ell(s)^2 ds.
\end{equation*}
Taking the limit $n\to +\infty$, we deduce from the monotone convergence theorem that $ \mathbb{E}\left[ \Delta_t^+\right] \le C\int_0^t \mathbb{E}[\Delta_s^+] ds$. Again, without loss of generality, we can assume that $\mathbb{E}\left[\int_0^t \Delta_s^+ds \right] <\infty$ because if this assumption does not hold, we can use a localization argument to return to a situation where it holds. Grönwall's lemma yields $\mathbb{E}\left[ \Delta_t^+\right] = 0$. Thus, $\mathbb{P}(X_t-Y_t\le 0)=1$ for all $t\ge 0$. We deduce that $\mathbb{P}(X_t\le Y_t, \forall t \ge 0) = 1$ by path continuity of $X$ and $Y$. 
\end{proof}

Endowed with these two lemmas, we are able to establish a result generalizing the one of \cite{nutz2023guyon}. 

\begin{proposition}\label{prop:positivity}
Under assumptions \ref{hyp:1.1}-\ref{hyp:1.5} and \ref{hyp:2.1}-\ref{hyp:2.3},
we have:
\begin{equation*}
  \sigma_t\ge \sigma_0\exp\left(\beta_1\int_0^t K_1(s,s) dW_s + h(t)-h(0) -\frac{1}{2}\beta_1^2\int_0^tK_1(s,s)^2ds  \right)
\end{equation*}
 for all $0\le t < \tau$ almost surely. 
\end{proposition}
\begin{proof}
  Let $0\le t<\tau$. There exists $n\ge 1$ such that $t< \hat{\tau}_n$ where $\hat{\tau}_n$ is defined in the second step of the proof of Theorem \ref{thm:main_result_1}. Because on $[0,\hat{\tau}_n)$, $R_{2,t}>1/n$ and because the square root function is twice continuously differentiable on $\mathbb{R}_+^*$, we can apply Itô's lemma, which combined with Lemma \ref{lem:differential_K} yields:
  \begin{equation*}
    d\sqrt{R_{2,t}} = \frac{1}{2\sqrt{R_{2,t}}}\left(\int_{-\Delta}^t \partial_t K_2(u,t)\sigma_u^2du \right)dt + \frac{1}{2\sqrt{R_{2,t}}}K_2(t,t)\sigma_t^2dt. 
  \end{equation*}
Noting that $\partial_t K_1(u,t)=h'(t)K_1(u,t)$, we also have by Lemma \ref{lem:differential_K}:
\begin{equation*}
  dR_{1,t} = h'(t)R_{1,t}dt + K_1(t,t)\sigma_t dW_t.
\end{equation*}
Thus, by definition of $\sigma_t$:
\begin{equation*}
  \sigma_t = \sigma_0+\int_0^t G_s ds + \int_0^t \gamma(s,\sigma_s) dW_s 
\end{equation*}
where $\gamma(s,x) = \beta_1K_1(s,s) x$ for $(s,x)\in \mathbb{R}_+\times\mathbb{R}$ and
\begin{equation*}
  \begin{aligned}
    G_t &=  \beta_1h'(t) R_{1,t} + \frac{\beta_2}{2\sqrt{R_{2,t}}}\left(\int_{-\Delta}^t \partial_t K_2(u,t)\sigma_u^2du \right)+ \frac{\beta_2}{2\sqrt{R_{2,t}}}K_2(t,t)\sigma_t^2 \\
    &= h'(t)\sigma_t -h'(t)\beta_0-h'(t)\beta_2 \sqrt{R_{2,t}} + \frac{\beta_2}{2\sqrt{R_{2,t}}}\left(\int_{-\Delta}^t \partial_t K_2(u,t)\sigma_u^2du \right)+ \frac{\beta_2}{2\sqrt{R_{2,t}}}K_2(t,t)\sigma_t^2 \\
    &=  h'(t)\sigma_t -h'(t)\beta_0+ \frac{\beta_2}{\sqrt{R_{2,t}}}\int_{-\Delta}^t\left(-h'(t)K_2(u,t)+\frac{1}{2}\partial_t K_2(u,t)\right)\sigma_u^2du +\frac{\beta_2}{2\sqrt{R_{2,t}}}K_2(t,t)\sigma_t^2.
  \end{aligned}
\end{equation*}
Using that $h'\le 0$, $\beta_0,\beta_2\ge 0$ and assumption \ref{hyp:2.3}, we have $G_t \ge h'(t)\sigma_t$. Moreover, $\gamma$ satisfies assumption \ref{hyp:sigma_reg} of Lemma \ref{lem:comparison_result} since $t\mapsto K_1(t,t)$ is locally square-integrable by assumption \ref{hyp:2.1}. Moreover, the function $b_X(t,x)=h'(t)x$ satisfies assumption \ref{hyp:b_ineq} in this lemma since $h'$ is non positive and such that $\inf_{t\ge 0}h'(t)>-\infty$ by assumption \ref{hyp:2.3}. According to Lemma \ref{lem:comparison_result}, we therefore have $\sigma_t \ge X_t$ for all $t\in [0,\tau)$ almost surely where $X$ is solution of the following SDE:
\begin{equation*}
  dX_t =h'(t) X_t dt + \beta_1K_1(t,t)X_t dW_t,\;X_0=\sigma_0.
\end{equation*}
 The solution of this SDE is $X_t=\sigma_0\exp\left(\beta_1\int_0^t K_1(s,s) dW_s + h(t)-h(0) -\frac{1}{2}\beta_1^2\int_0^tK_1(s,s)^2ds  \right)$ which concludes the proof. 
\end{proof}

We conclude the proof of Theorem \ref{thm:main_result_2} by showing that condition \ref{hyp:1.6} is implied by the conditions of Theorem \ref{thm:main_result_2} so that we can apply Theorem \ref{thm:main_result_1} to get the global existence and uniqueness of a solution $(R_1,R_2)$ to Equation (\ref{eq:r_dynamics}). 

\begin{lemma}\label{lem:condition_1.6}
  Assumption \ref{hyp:1.6} holds if assumptions \ref{hyp:2.2}, \ref{hyp:2.3} and \ref{hyp:2.4} are satisfied. 
\end{lemma}
\begin{proof}
Let $T>0$ and $t\in[0,T]$. As already observed in Remark \ref{rk:cutoff}, assumption \ref{hyp:2.3} implies that:
  \begin{equation*}
    K_2(s,t) \ge K_2(s,0) e^{2\int_0^t h'(u)du}\quad \forall s \in (-\Delta,0]. 
  \end{equation*}
  Thus,
  \begin{equation*}
    \begin{aligned}
    g_2(t)&:= \int_{-\Delta}^0 K_2(s,t) \left(\beta_0+\beta_1r_{1,s}+\beta_2\sqrt{r_{2,s}}\right)^2 ds \\
    &\ge e^{2\int_0^t h'(u)du }\underbrace{\int_{-\Delta}^0 K_2(s,0) \left(\beta_0+\beta_1r_{1,s}+\beta_2\sqrt{r_{2,s}}\right)^2ds}_{g_2(0)}.
    \end{aligned}
  \end{equation*}
  Hence,
  \begin{equation*}
    \inf_{t\in [0,T]} g_2(t) \ge e^{2\int_0^T h'(u)du} g_2(0) >0
  \end{equation*}
  by assumptions \ref{hyp:2.2} and \ref{hyp:2.3}. 
\end{proof}


\bibliographystyle{abbrvnat}
\bibliography{bibli}
\end{document}